\documentclass[english,nolineno]{socg-lipics-v2019}

\usepackage{microtype}

\usepackage{amsmath,amsfonts,amssymb,amsthm}
\usepackage{graphicx}
\usepackage{url}
\usepackage{todonotes}
\usepackage{color}
\usepackage{xspace}

\theoremstyle{plain}
\newtheorem{obs}{Observation}[section]
\newtheorem{lem}{Lemma}[section]

\newtheorem{prop}{Proposition}[section]

\usepackage{algorithm}
\usepackage{algorithmic}

\usepackage{amsopn}
\DeclareMathOperator{\si}{\sf sign}
\DeclareMathOperator{\red}{\sf reduction}
\DeclareMathOperator{\con}{\sf convex}
\DeclareMathOperator{\reflex}{\sf reflex}

\DeclareMathOperator{\first}{\sf first}
\DeclareMathOperator{\last}{\sf last}

\DeclareMathOperator{\nil}{\sf nil}
\DeclareMathOperator{\og}{\sf og}
\DeclareMathOperator{\ig}{\sf ig}
\DeclareMathOperator{\how}{\sf how}
\DeclareMathOperator{\bij}{\sf bij}

\newcommand{\OO}{{\mathcal O}}

\newcommand{\Yes}{{\sc Yes}}
\newcommand{\No}{{\sc No}}
\newcommand{\FPT}{\textrm{\textup{FPT}}\xspace}
\newcommand{\WO}{\textrm{\textup{W[1]}}}
\newcommand{\WOH}{\textrm{\textup{W[1]-hard}}\xspace}

\newcommand{\NP}{\textrm{\textup{NP}}\xspace}
\newcommand{\ETH}{\textrm{\textup{ETH}}\xspace}
\newcommand{\NPH}{\textrm{\textup{NP-hard}}\xspace}

\newcommand{\APXH}{\textrm{\textup{APX-hard}}\xspace}

\newcommand{\Essn}{{\sf Ess}}
\newcommand{\asg}{\alpha}
\newcommand{\SC}[1]{\mathscr{#1}} 
\newcommand{\C}[1]{\mathcal{#1}} 
\newcommand{\what}{\widehat} 
 
\newcommand{\rtimedcsp}{\C{O}((|X| + |C|) \cdot N)}
\newcommand{\rtimetsat}{\C{O}(n+m)}

\newcommand{\myparagraph}[1]{\smallskip\noindent{\textbf{\sffamily #1}}}

\title{The Parameterized Complexity of Guarding Almost Convex Polygons\footnote{A preliminary version of this article is accepted for publication at the 36th International Symposium on Computational Geometry (SoCG 2020).}}

\titlerunning{The Parameterized Complexity of Guarding Almost Convex Polygons}

\author{Akanksha Agrawal}{Ben-Gurion University, Beersheba, Israel}{agrawal@post.bgu.ac.il}{https://orcid.org/0000-0002-0656-7572}{the PBC Fellowship Program for Outstanding Post-Doctoral Researchers from China and India.}

\author{Kristine V.K. Knudsen}{University of Bergen, Bergen, Norway}{kristine.knudsen@ii.uib.no}{}{}

\author{Daniel Lokshtanov}{University of California, Santa Barbara, USA}{daniello@ucsb.edu}{}{European Research Council (ERC) under the European Union's Horizon 2020
research and innovation programme (no. 715744), and United States - Israel Binational
Science Foundation (no. 2018302).}

\author{Saket Saurabh}{The Institute of Mathematical Sciences, HBNI, Chennai, India}{saket@imsc.res.in}{}{European Research Council (ERC) under the European Union's Horizon 2020 research and innovation programme (no. 819416), and Swarnajayanti Fellowship (no. DST/SJF/MSA01/2017-18).}

\author{Meirav Zehavi}{Ben-Gurion University, Beersheba, Israel}{meiravze@bgu.ac.il}{https://orcid.org/0000-0002-3636-5322}{Israel Science Foundation grant no. 1176/18, and United States - Israel Binational
Science Foundation (no. 2018302).}

\authorrunning{A. Agrawal, K. Knudsen, D. Lokshtanov, S. Saurabh, M. Zehavi}

\Copyright{Akanksha Agrawal, Kristine V.K. Knudsen, Daniel Lokshtanov, Saket Saurabh and Meirav Zehavi}



\acknowledgements{We thank anonymous reviewers for helpful comments that improved and simplified the paper.}

\EventEditors{}
\EventNoEds{0}
\EventLongTitle{}
\EventShortTitle{}
\EventAcronym{}
\EventYear{}
\EventDate{}
\EventLocation{}
\EventLogo{}
\SeriesVolume{}
\ArticleNo{}

\begin{document}

\maketitle
\begin{abstract}
The {\sc Art Gallery} problem is a fundamental visibility problem in Computational Geometry. The input consists of a simple polygon $P$, (possibly infinite) sets $G$ and $C$ of points within $P$, and an integer $k$; the task is to decide if at most $k$ guards can be placed on points in $G$ so that every point in $C$ is visible to at least one guard. In the classic formulation of {\sc Art Gallery}, $G$ and $C$ consist of all the points within $P$. Other well-known variants restrict $G$ and $C$ to consist either of all the points on the boundary of $P$ or of all the vertices of~$P$. Recently, three new important discoveries were made: the above mentioned variants of {\sc Art Gallery} are all \WOH\ with respect to $k$ [Bonnet and Miltzow, ESA'16], the classic variant has an $\OO(\log k)$-approximation algorithm [Bonnet and Miltzow, SoCG'17], and it may require irrational guards [Abrahamsen et al., SoCG'17]. Building upon the third result, the classic variant and the case where $G$ consists only of all the points on the boundary of $P$ were both shown to be $\exists\mathbb{R}$-complete~[Abrahamsen et al., STOC'18]. Even when both $G$ and $C$ consist only of all the points on the boundary of $P$, the problem is not known to be in \NP. 

Given the first discovery, the following question was posed by Giannopoulos [Lorentz Center Workshop, 2016]: Is {\sc Art Gallery} \FPT\ with respect to $r$, the number of reflex vertices? In light of the developments above, we focus on the variant where $G$ and $C$ consist of all the vertices of $P$, called {\sc Vertex-Vertex Art Gallery}. Apart from being a variant of {\sc Art Gallery}, this case can also be viewed as  the classic {\sc Dominating Set} problem in the visibility graph of a polygon. In this article, we show that the answer to the question by Giannopoulos is {\em positive}: {\sc Vertex-Vertex Art Gallery} is solvable in time $r^{\OO(r^2)}n^{\OO(1)}$. Furthermore,  our approach extends to assert that {\sc Vertex-Boundary Art Gallery} and {\sc Boundary-Vertex Art Gallery} are both \FPT\ as well. To this end, we utilize structural properties of ``almost convex polygons'' to present a two-stage reduction from {\sc Vertex-Vertex Art Gallery} to a new constraint satisfaction problem (whose solution is also provided in this paper)  where constraints have arity $2$ and involve monotone functions.
\end{abstract}

\section{Introduction}\label{sec:intro}

Given a {\em simple} polygon $P$ on $n$ vertices, two points $x$ and $y$ within $P$ are {\em visible} to each other if the line segment between $x$ and $y$ is contained in $P$. Accordingly, a set $S$ of points within $P$ is said to {\em guard} another set $Q$ of points within $P$ if, for every point  $q\in Q$, there is some point $s\in S$ such that $q$ and $s$ are visible to each other. The computational problem that arises from this notion is loosely termed the {\sc Art Gallery} problem.  In its general formulation, the input consists of a simple polygon $P$, possibly infinite sets $G$ and $C$ of points within $P$, and a non-negative integer $k$. The task is to decide whether at most $k$ guards can be placed on points in $G$ so that every point in $C$ is visible to at least one guard. The most well-known cases of {\sc Art Gallery} are identified as follows: the {\sc X-Y Art Gallery} problem is the {\sc Art Gallery} problem where $G$ is the set of all points within $P$ (if {\sc X}={\sc Point}), all boundary points of $P$ (if {\sc X}={\sc Boundary}), or all vertices of $P$ (if {\sc X}={\sc Vertex}), and $C$ is defined analogously with respect to {\sc Y}. The classic variant of {\sc Art Gallery} is the {\sc Point-Point Art Gallery} problem. Nevertheless, all variants where {\sc X}={\sc Vertex} or {\sc Y}={\sc Point} received attention in the literature.\footnote{The {\sc X-Y Art Gallery} problem, for any {\sc X,Y} $\in\{${\sc Point, Boundary, Vertex}$\}$, is often loosely termed the {\sc Art Gallery} problem. For example, in the survey of open problems by Ghosh and Goswami~\cite{DBLP:journals/csur/GhoshG13}, the term {\sc Art Gallery} problem refers to the {\sc Vertex-Vertex Art Gallery} problem.} In particular, {\sc Vertex-Vertex Art Gallery} is equivalent to the classic {\sc Dominating Set} problem in the visibility graph~of~a~polygon.

The {\sc Art Gallery} problem is a fundamental visibility problem in Discrete and Computational Geometry, which was extensively studied from both combinatorial and algorithmic viewpoints. The problem was first proposed by Victor Klee in 1973, which prompted a flurry of results~\cite[page 1]{ArtGalBook1}. The main combinatorial question posed by Klee was {\em how many guards are sufficient to see every point of the interior of an $n$-vertex simple polygon?} Chv\'{a}tal~\cite{Chvatal75} showed in 1975 that  $\lfloor \frac{n}{3}\rfloor$ guards are always sufficient and sometimes necessary for any $n$-vertex simple polygon (see~\cite{Fisk78a} for a simpler proof by Fisk). After this, many variants of the {\sc Art Gallery} problem, based on different definitions of visibility, restricted classes of polygons, different shapes of guards, and mobility of guards, have been defined and analyzed.
A book~\cite{ArtGalBook1} and several extensive surveys and book chapters were dedicated to {\sc Art Gallery} and its variants (see, e.g.,~\cite{RezendeSFHKT16,shermer1992recent,urrutia2000art}). In this article, our main proof states that the {\sc Vertex-Vertex Art Gallery} problem is {\em fixed-parameter tractable (\FPT)} parameterized by $r$, the number of reflex vertices of $P$.  Additionally, we show that both {\sc Vertex-Boundary Art Gallery} and {\sc Boundary-Vertex Art Gallery}~are~\FPT\ with respect to the number of reflex vertices~as~well.

\subparagraph{1.1. Background: Related Algorithmic Works.}
In what follows, we focus only on algorithmic works on {\sc X-Y Art Gallery} for {\sc X,Y}$\in\{${\sc Point,Boundary,Vertex}$\}$. 

\medskip
\noindent{\bf Hardness.} In 1983, O'Rourke and Supowit~\cite{DBLP:journals/tit/ORourkeS83} proved that {\sc Point-Point Art Gallery} is \NPH\ if the polygon can contain holes. The requirement to allow holes was lifted shortly afterwards~\cite{phdthesisAgrawal}. 
In 1986, Lee and Lin~\cite{DBLP:journals/tit/LeeL86} showed that {\sc Vertex-Point Art Gallery} is \NPH. This result extends to  {\sc Vertex-Vertex Art Gallery} and {\sc Vertex-Boundary Art Gallery}. Later, numerous other restricted cases were shown to be \NPH\ as well. For example, \NP-hardness was established for orthogonal polygons by Katz and Roisman \cite{TerrainApproxOrtho} and Schuchardt and Hecker \cite{OrthoPolygonNPhard}. We remark that the reductions that show that {\sc X-Y Art Gallery} (for {\sc X,Y} $\in\{${\sc Point, Boundary, Vertex}$\}$) is \NPH\ also imply that these cases cannot be solved in time $2^{o(n)}$ under the Exponential-Time Hypothesis (\ETH).

While it has long been known that even very restricted cases of {\sc Art Gallery} are \NPH, the inclusion of {\sc X-Y Art Gallery}, for {\sc X,Y} $\in\{${\sc Point, Boundary}$\}$, in \NP\ remained open. (When {\sc X}={\sc Vertex}, the problem is clearly in \NP.) In 2017, Abrahamsen et al.~\cite{DBLP:conf/compgeom/AbrahamsenAM17} began to reveal the reasons behind this discrepancy for the {\sc Point-Point Art Gallery} problem: they showed that {\em exact} solutions to this problem sometimes require placement of guards on points with {\em irrational} coordinates. Shortly afterwards, they extended this discovery to prove that {\sc Point-Point Art Gallery} and {\sc Boundary-Point Art Gallery} are $\exists\mathbb{R}$-complete~\cite{DBLP:conf/stoc/AbrahamsenAM18}. Roughly speaking, this result means that
{\em (i)} any system of polynomial equations over the real numbers can be encoded as an instance of {\sc Point/Boundary-Point Art Gallery}, and {\em (ii)} these problems are not in the complexity class \NP\ unless \NP\ = $\exists\mathbb{R}$.

\medskip
\noindent{\bf Approximation Algorithms.} The {\sc Art Gallery} problem has been extensively studied from the viewpoint of approximation algorithms~\cite{EfratH06,DeshpandeKDS07,Ghosh10,KING2013219,DBLP:journals/dcg/KingK11,KrohnN13,Kirkpatrick15,BonnetM17,Bhattacharya2017ApproximabilityOG,DBLP:journals/corr/abs-1712-05492,DBLP:journals/corr/abs-1803-02160} (this list is not comprehensive). Most of these approximation algorithms are based on the fact that the range space defined by the visibility regions has bounded VC-dimension for simple polygons~\cite{GilbersK14,kalai1997guarding,valtr1998guarding}, which facilitates the usage of the algorithmic ideas of Clarkson~\cite{BronnimannG95,Clarkson93}. The current state-of-the-art is as follows. For the {\sc Boundary-Point Art Gallery} problem, King and Kirkpatrick~\cite{DBLP:journals/dcg/KingK11} gave a factor $\OO(\log\log {\sf OPT})$ approximation algorithm. For the {\sc Point-Point Art Gallery} problem, Bonnet and Miltzow~\cite{BonnetM17} gave a factor $\OO(\log {\sf OPT})$ approximation algorithm. 
Very recently, in a yet unpublished work, Bhattacharya et al.~\cite{DBLP:journals/corr/abs-1712-05492} reported a breakthrough: they designed an 18-approximation algorithm for {\sc Vertex-Vertex Art Gallery}, a (slightly slower) 18-approximation algorithm for {\sc Vertex-Boundary Art Gallery}, and a 27-approximation algorithm for {\sc Vertex-Point Art Gallery}. For all of these three variants, the existence of a constant-factor approximation algorithm has been a longstanding open problem, conjectured to be true already in 1987 by Ghosh~\cite{gosh87,Ghosh10,DBLP:journals/csur/GhoshG13}.
The existence of a constant-factor approximation algorithm for {\sc Point-Point Art Gallery} (or even {\sc Boundary-Boundary Art Gallery} or {\sc Boundary-Point Art Gallery}) remains a major open problem. On the negative side, all of these variants are known to be \APXH~\cite{DBLP:conf/cccg/EidenbenzSW98,DBLP:journals/algorithmica/EidenbenzSW01}. However,  restricted classes of polygons, such as weakly-visible polygons~\cite{DBLP:journals/corr/abs-1803-02160}, give rise to a PTAS.

\medskip
\noindent{\bf Exact Algorithms.} For an $n$-vertex polygon $P$, one can efficiently find a set of $\lfloor \frac{n}{3}\rfloor$ vertices that guard all points within $P$, matching  Chv\'{a}tal's upper bound~\cite{Chvatal75}. Specifically, Avis and Toussaint~\cite{DBLP:journals/pr/AvisT81} presented an $\OO(n \log n)$-time divide-and-conquer algorithm for this task. Later, Kooshesh and Moret~\cite{DBLP:journals/pr/KoosheshM92} gave a linear-time algorithm based on Fisk's short proof~\cite{Fisk78a}.
However, when we seek an optimal solution, the situation is much more complicated. The first exact algorithm for {\sc Point-Point Art Gallery} was published in 2002 in the conference version of a paper by Efrat and Har-Peled~\cite{EfratH06}. They attribute the result to Micha Sharir. Before that time, the problem was not even known to be decidable. The algorithm computes a formula in the first order theory of the reals corresponding to the art gallery instance (with both existential and universal quantifiers), and employs algebraic methods such as the techniques provided by Basu et al.~\cite{DBLP:journals/jacm/BasuPR96}, to decide if the formula is true. Given that {\sc Point-Point Art Gallery} is $\exists\mathbb{R}$-complete~\cite{DBLP:conf/stoc/AbrahamsenAM18}, it might not be possible to avoid the use of this powerful machinery. However, even for the cases where {\sc X}={\sc Vertex}, the situation is quite grim; we are not aware of {\em exact} algorithms that achieve substantially better time complexity bounds than brute-force.
Nevertheless, over the years, exact algorithms that perform well in practice 
were developed. For example, see \cite{DBLP:conf/compgeom/BorrmannRSFFKNST13,RezendeSFHKT16,DBLP:journals/itor/CoutoRS11}.

\medskip
\noindent{\bf Parameterized Complexity.} 
Two years ago, Bonnet and Miltzow~\cite{ArtGalW1Hard} showed that {\sc Vertex-Point Art Gallery} and {\sc Point-Point Art Gallery} are \WOH\ with respect to the {\em solution size}, $k$. 
 With straightforward adaptations, their results extend to most of the known variants of the problem, including {\sc Vertex-Vertex Art Gallery}. Thus, {\em the classic parameterization by solution size leads to a dead-end}. However, this does not rule out the existence of \FPT\ algorithms for non-trivial structural parametrizations. We refer to the nice surveys by  Niedermeier on the art of parameterizations~\cite{Niedermeier04,Niedermeier10}.

\subparagraph{1.2. Giannopoulos's Parameterization and Our Contribution.}\label{sec:ourContribution}
In light of the \WO-hardness result by Bonnet and Miltzow~\cite{ArtGalW1Hard}, Giannopoulos~\cite{TerrainQ} proposed to parameterize the {\sc Art Gallery} problem by the number $r$ of reflex vertices of the input polygon $P$. Specifically, Giannopoulos~\cite{TerrainQ} posed the following open problem:
 {\em ``Guarding simple polygons has been recently shown to be \WO-hard w.r.t.~the number~of~(vertex or edge) guards. Is the problem \FPT\ w.r.t.~the number of reflex vertices of the polygon?''}
The motivation behind this proposal is encapsulated by the following well-known proposition, see~\cite[Sections 2.5-2.6]{ArtGalBook1}.

\begin{prop}[Folklore]\label{prop:reflex}
For any polygon $P$, the set of reflex vertices of $P$ guards the set of all points within $P$.
\end{prop}

\begin{figure}[t]
\centering
\fbox{\includegraphics[scale=0.8]{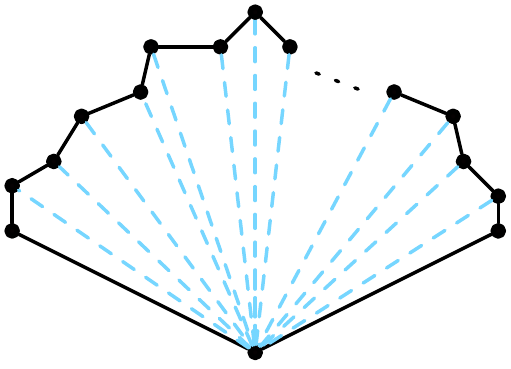}}
\caption{The solution size $k=1$, yet the number of reflex vertices $r$ is arbitrarily large.}\label{fig:largeR}
\end{figure}

That is, the minimum number $k$ of guards needed (for any of the cases of {\sc Art Gallery}) is upper bounded by the number of reflex vertices $r$.  Clearly, $k$ can be arbitrarily smaller than $r$ (see Fig.~\ref{fig:largeR}). 
Our main result is that the  {\sc Vertex-Vertex Art Gallery} problem is \FPT  parameterized by $r$. This implies that guarding the vertex set of ``almost convex polygons'' is easy. In particular, whenever $r^2\log r=\OO(\log n)$, the problem is solvable~in~polynomial~time.

\begin{theorem}\label{mainthm:agpreflex}
{\sc Vertex-Vertex Art Gallery} is \FPT\ parameterized by $r$, the number of reflex vertices. In particular, it admits an algorithm with running time $r^{\OO(r^2)}n^{\OO(1)}$.
\end{theorem}

A few remarks are in place. First, our result extends (with straightforward adaptation) to the most general discrete annotated case of {\sc Art Gallery} where $G$ and $C$ are each a subset of the vertex set of the polygon, which can include points where the interior angle is of 180 degrees. Consequently, a simple discretization procedure shows that {\sc Vertex-Boundary Art Gallery} and {\sc Boundary-Vertex Art Gallery} are both \FPT\ parameterized by $r$ as well. However, we do not know how to handle {\sc Vertex-Point Art Gallery} and {\sc Point-Vertex Art Gallery}; determining whether these variants are \FPT\ with respect to $r$ remains open.
 Second, for variants where  both {\sc X} $\neq$ {\sc Vertex} and {\sc Y} $\neq$ {\sc Vertex}, the design of {\em exact} algorithms poses extremely difficult challenges. As discussed earlier, these cases are not even known to be in \NP; in particular, {\sc Point-Point Art gallery} is $\exists\mathbb{R}$-hard~\cite{DBLP:conf/stoc/AbrahamsenAM18}. Moreover, there is only one known exact algorithm that resolves these cases and it employs extremely powerful machinery (as a black box), not known to be avoidable. Third, note that our result is among very few {\em positive} results that concern {\em optimal} solutions to (any case of) {\sc Art Gallery}.

Along the way to establish our main result, we prove that a constraint satisfaction problem called {\sc Monotone 2-CSP} is solvable in polynomial time. This result might be of independent interest. Informally, in {\sc Monotone 2-CSP}, we are given $k$ variables and $m$ constraints. Each constraint is of the form $[x \si f(x')]$ where $x$ and $x'$ are variables, $\si\in\{\leq,\geq\}$, and $f$ is a {\em monotone} function. The objective is to assign an integer from $\{0,1,\ldots,N\}$ to each variable so that all of the constraints will be satisfied. For this problem, we develop a surprisingly simple algorithm based on a reduction to {\sc 2-CNF-SAT}.

\begin{theorem}\label{thm:csp}
{\sc Monotone 2-CSP} is solvable in polynomial time.
\end{theorem}

Essentially, the main technical component of our work is an exponential-time reduction that creates an exponential (in $r$) number of instances of {\sc Monotone 2-CSP} so that the original instance is a \Yes -instance if and only if at least one of the instances of  {\sc Monotone 2-CSP} is a \Yes -instance.  Our reduction is done in two stages due to its structural complexity. In the first stage of the reduction, we aim to make ``guesses'' that determine the relations between the ``elements'' of the problem (that are the ``critical'' visibility relations in our case) and thereby elucidate and further binarize them (which, in our case, is required to impose order on guards). This part requires exponential time (given that there are exponentially many guesses) and captures the ``NP-hardness'' of the problem. Then, the second stage of the reduction is to translate each guess into an instance of {\sc Monotone 2-CSP}. This part, while requiring polynomial time, relies on a highly non-trivial problem-specific insight---specifically, here we need to assert that the relations considered earlier can be encoded by constraints that are not only binary, but that the functions they involve are {\em monotone}. We strongly believe that our approach can be proven fruitful to resolve the parameterized complexity of other problems of discrete~geometric~flavour.

\subparagraph{1.3 Our Methods.}
The proof of Theorem \ref{mainthm:agpreflex} consists of four components (see Fig.~\ref{fig:methods}). The first component (in Section \ref{sec:structureAG}) establishes several structural claims regarding monotone properties of visibility in polygons.

\begin{figure}[t]
\centering
\fbox{\includegraphics[scale=0.75]{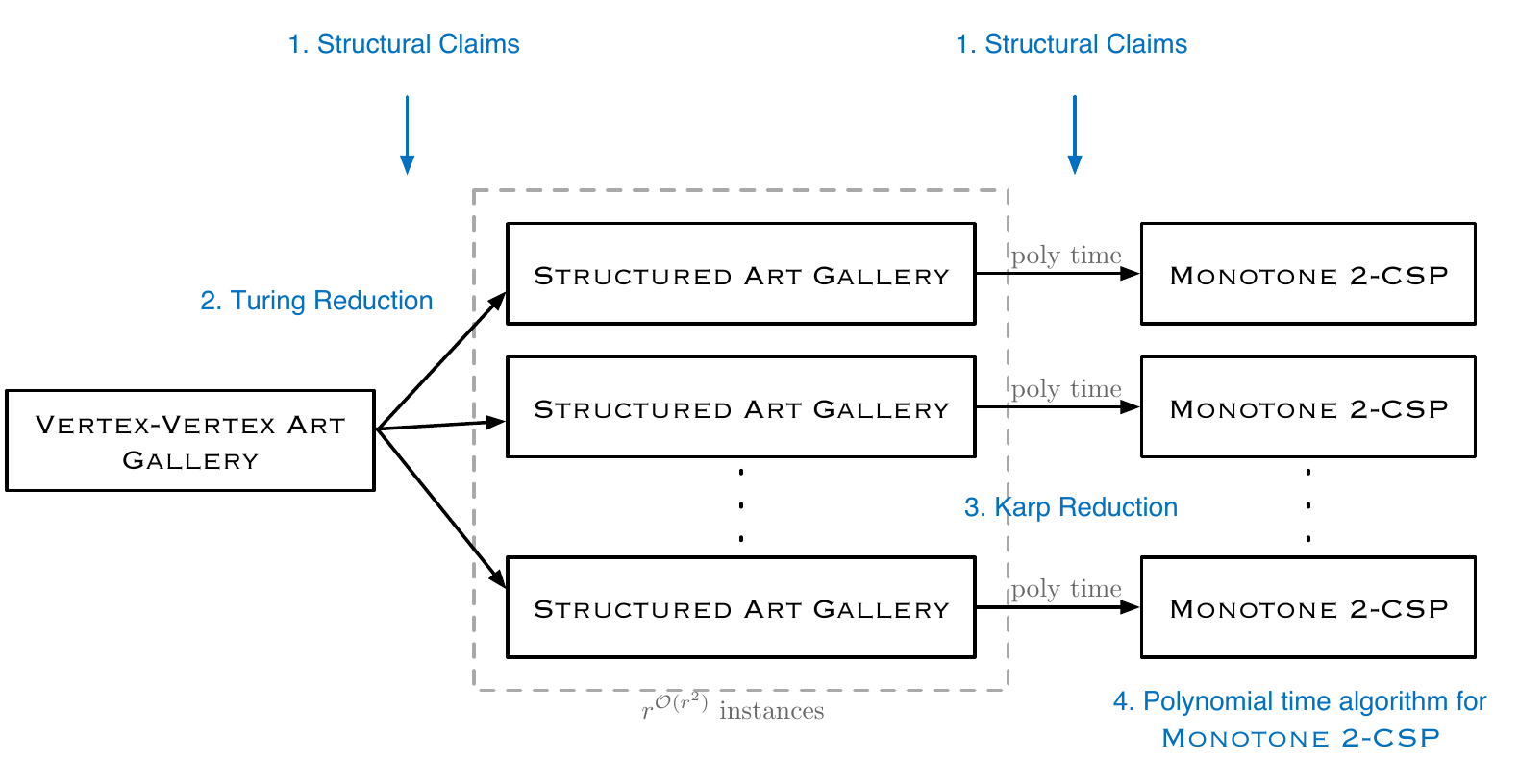}}
\caption{The four components of our proof.}\label{fig:methods}
\end{figure}

%
Informally, we order the vertices of the polygon according to their appearance on the boundary, and consider each sequence between two reflex vertices to be a ``convex region''. Then, we argue that for every pair of convex regions, as we ``move along'' one of them, the (index of the) first vertex in the other region that we see either never becomes smaller or never becomes larger. Symmetrically, this claim also holds for the last visible vertices that we encounter. In addition, we argue that if a vertex sees some two vertices in a convex region, then it also sees all vertices in between these two vertices.

Our second component (in Section \ref{sec:structured}) is a Turing reduction to an intermediate problem that we term {\sc Structured Art Gallery}. Roughly speaking, in this problem, each convex region ``announces'' how many guards it will contain, and how many guards are necessary to see it completely. In addition, it announces that a prefix of the sequence that forms this region will be guarded by, say, ``the $i^{th}$ guard to be placed on region $C$'', then the following subsequence will be guarded by, say, ``the $j^{th}$ guard to be placed on region $C'$'', and so on, until it announces how a suffix of it is to be guarded. We stress that the identity of what is ``the $i^{th}$ guard to be placed on region $C$'', or what is ``the $j^{th}$ guard to be placed on region $C'$'', are of course not known, and should be discovered. Moreover, even the division into subsequences is not known. In {\sc Structured Art Gallery}, we only focus on solutions that are of the above form. We utilize our second component not only to impose these additional conditions, but also to begin the transition from the usage of visibility-based conditions to function-based constraints. Specifically, functions called $\first$ and $\last$ will encode, for any vertex $v$ and convex region $C$, the first and last vertices in $C$ visible to $v$. To argue that such simple functions encode all necessary information concerning visibility, we make use of the structural claims~established~earlier. 

Our third component (in Section \ref{sec:reduction})
 is a Karp reduction from {\sc Structured Art Gallery} to the constraint satisfaction problem, {\sc Monotone 2-CSP}, discussed in Section \ref{sec:ourContribution}. This is the part of the proof that most critically relies on all of the structural claims established earlier. Here, we need to translate the constraints imposed by {\sc Structured Art Gallery} into constraints that comply with the very restricted form of an instance of {\sc Monotone 2-CSP}, namely, being monotone and involving only two variables. We remark that if one removes the requirement of monotonicity, or allows each constraint to consist of more variables, then the problem can be easily shown to encode {\sc Clique} and hence become \WO-hard (see Section \ref{sec:reduction}). The translation entails a non-trivial analysis to ensure that all functions are indeed monotone. Specifically, each convex region requires its own set of tailored functions to enforce some relationships between the (unknown) guards it announced to contain and the (unknown) subsequences that these guards are supposed to see.
In a sense, our first three components extract the algebraic essence of the {\sc Vertex-Vertex Art Gallery} problem: by identifying monotone properties and making guesses to ensure binary dependencies between solution elements, the problem is encoded by a~restricted~constraint~satisfaction~problem. 

Lastly, our fourth component is a relatively simple polynomial-time algorithm for {\sc Monotone 2-CSP} (see Theorem \ref{thm:csp}), given in Appendix \ref{sec:algoCSP}, based on a reduction to {\sc 2-CNF-SAT}. Essentially, the crux is {\em not} to encode every pair of a variable of {\sc Monotone 2-CSP} and a potential value for it as a variable of {\sc 2-CNF-SAT} that signifies equality, because then, although the functions become easily encodable in the language of {\sc 2-CNF-SAT}, it is unclear how to ensure that each variable of {\sc Monotone 2-CSP} will be in exactly one pair that corresponds to a variable assigned true when satisfying the 2-CNF-SAT formula. Indeed, the naive approach seems futile, because it does not exploit the monotonicity of the input functions. Instead, for each pair of a variable of {\sc Monotone 2-CSP} and a potential value for it with the exception of $0$, we introduce a variable of {\sc 2-CNF-SAT} signifying that the variable is assigned {\em at least} the value in the pair. The assignment of value $0$ is implicitly encoded by the negation of pairs with the value $1$. Then, we can ensure that each variable is assigned exactly one value (when translating a truth assignment for the {\sc 2-CNF-SAT} instance we created back into an assignment for the {\sc Monotone 2-CSP} input instance), and by relying on the monotonicty of the input functions, we are able to encode them correctly in the language of {\sc 2-CNF-SAT}.

For notational clarity, we describe our proof for  {\sc Vertex-Vertex Art Gallery}. However, all arguments extend in a straightforward manner to solve the annotated generalization of {\sc Vertex-Vertex Art Gallery} where $G$ and $C$ are each a subset of the vertex set of the polygon. Then, simple discretization procedures yield the positive resolution of the parameterized complexity also of {\sc Vertex-Boundary Art Gallery} and {\sc Boundary-Vertex Art Gallery} (see Section~\ref{sec:discretization}).

\section{Preliminaries}\label{sec:prelims}

We use standard  terminology from the book of Diestel~\cite{DiestelBook}. 
With the exception of the Introduction, the abbreviation {\sc Art Gallery} refers to {\sc Vertex-Vertex Art Gallery}.

\medskip
\noindent{\bf Polygons.} A {\em simple polygon} $P$ is a flat shape consisting of $n$ straight, {\em non-intersecting} line segments that are joined pair-wise to form a closed path. The line segments that make up a polygon, called {\em edges}, meet only at their endpoints, called {\em vertices}. Any polygon can be modeled by a graph $P=(V,E)$ with $V=\{1,2,\ldots,n\}$ and $E=\{\{i,i+1\}\}: i\in \{1,\ldots,n-1\}\}\cup\{\{n,1\}\}$ where every vertex $i\in V$ is associated with a point $(x_i,y_i)$ on the Euclidean plane. 
A simple polygon $P$ encloses a region, called its {\em interior}, that has a measurable area. We consider the boundary of $P$ as part of its interior. A vertex $i\in V$ is a {\em reflex} (resp.~{\em convex}) vertex if the interior angle of $P$ at $i$ is larger (resp.~smaller) than 180 degrees. If $i\in V$ is not a reflex vertex, then either $i$ is a convex vertex or the interior angle of $P$ at $i$ is exactly 180 degrees. We slightly abuse notation and refer to all non-reflex vertices as convex vertices. We denote the set of reflex vertices of $P$ by $\reflex(P)$, and the set of convex vertices of $P$ by $\con(P)$.
A {\em convex polygon} $P$ is a simple polygon such that for every two points $p$ and $q$ on the boundary (or interior) of $P$, no point of the line segment $\overline{pq}$ is strictly outside $P$. In a convex polygon, all interior angles are less than or equal to 180 degrees, while in a strictly convex polygon all interior angles are less than 180 degrees. Given a non-convex polygon $P=(V,E)$, we suppose w.l.o.g.~that $1\in V$ is a reflex~vertex.

\medskip
\noindent{\bf Visibility.} Let $P=(V,E)$ be a simple polygon. We say that a point $p$ {\em sees} (or is {\em visible} to) a point $q$ if every point of the  line segment $\overline{pq}$ belongs to the interior (including the boundary) of $P$. More generally, a set of points $S$ {\em sees} a set of points $Q$ if every point in $Q$ is seen by at least one point in $S$. Note that if a point $p$ sees a point $q$, then the point $q$ sees the point $p$ as well. Moreover, a vertex $v\in V$ necessarily sees itself and its two neighbors in $P$.
%
The definition of a convex polygon asserts that the following observation holds.

\begin{obs}\label{obs:convexPolygon}
Any point within a convex polygon $P$ sees all points within~$P$. 
\end{obs}

\medskip
\noindent{\bf Parameterized Complexity.} Every instance of a parameterized problem is accompanied by a parameter $k$. A parameterized problem $\Pi$ is {\em fixed-parameter tractable (\FPT)} if there is an algorithm that, given an instance $(I, k)$ of $\Pi$, solves it in time $f(k)|I|^{\OO(1)}$ where $f$ is some computable function of $k$. Under reasonable complexity-theoretic assumptions, there are parameterized problems (such as  \WOH\ problems) that are not \FPT. For more information, we refer the reader to monographs such as \cite{CyganFKLMPPS15,DowneyF13}.

\section{Algorithm for Art Gallery}\label{sec:art}
In this section, we prove that {\sc Art Gallery} is \FPT\ with respect to $r$, the number of reflex vertices, by developing an algorithm with running time $2^{\OO(r^2\log r)}n^{\OO(1)}$.
We first present structural claims that exhibit the monotone way in which vertices in a so called ``convex region'' see vertices in another such region (Section \ref{sec:structureAG}). Then, we present a Turing reduction from {\sc Art Gallery} to a  problem called {\sc Structured Art Gallery} (Section \ref{sec:structured}).
 Next, based on the claims in Section \ref{sec:structureAG}, we present our main reduction, which translates {\sc Structured Art Gallery} to {\sc Monotone 2-CSP} (Section \ref{sec:reduction}). By developing an algorithm for {\sc Monotone 2-CSP} (Appendix \ref{sec:algoCSP}), we conclude the proof.

\subsection{Simple Structural Claims}\label{sec:structureAG}

\begin{figure}[t]
\centering
\fbox{\includegraphics[scale=0.8]{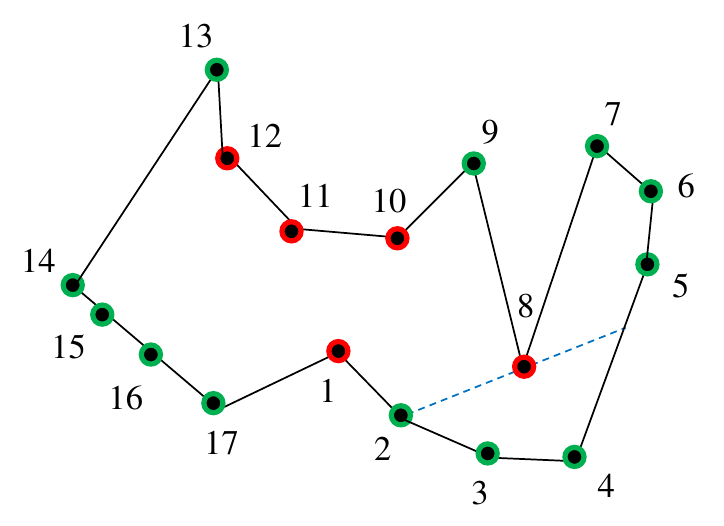}}
\caption{A simple polygon with three maximal convex regions: $[2,7]$, $[9]$ and $[13,17]$. Although $2,5\in[2,7]$ belong to the same convex region, they do not see each other.}\label{fig:convexRegions}
\end{figure}

We begin our analysis with the definition of a subsequence of vertices termed a convex region, illustrated in Fig.~\ref{fig:convexRegions}. 
Below, $j+1$ for $j=n$ refers to $1$. Because we assumed that vertex $1$ of any non-convex polygon is a reflex vertex, any convex region $[i,j]$ satisfies $i\neq 1$.

\begin{definition}\label{def:maximal-convex-region-vvag}
Let $P=(V,E)$ be a simple polygon. A non-empty set of vertices $[i,j]=\{i,i+1,\ldots,j\}$ is a {\em convex region} of $P$ if all the vertices in $[i,j]$ are convex. In addition, if $i-1\geq 1$ and $j+1$ are reflex vertices, 
then $[i,j]$ is a {\em maximal} convex region.
\end{definition}

In what follows, we would like to argue that for every two (not necessarily distinct) convex regions, one convex region sees the other in a manner that is ``monotone'' for each ``orientation'' in which we traverse these regions. To formalize this, we make use of the following notation, illustrated in Fig.~\ref{fig:view}. For a polygon $P=(V,E)$, a convex region $[i,j]$ of $P$ and a vertex $v\in V$, denote the smallest and largest vertices in $[i,j]$ that are seen by $v$ by $\first(v,[i,j])$ and $\last(v,[i,j])$, respectively. If $v$ sees no vertex in $[i,j]$, define $\first(v,[i,j])=\last(v,[i,j])=\nil$. 
Accordingly, we define two types of monotone views. First, we address the orientation corresponding to $\first$ (see Fig.~\ref{fig:view}). Roughly speaking, we say that the way a convex region $[i,j]$ views a  convex region $[i',j']$ is, say, non-decreasing with respect to $\first$, if when we traverse $[i,j]$ from $i$ to $j$ and consider the first vertices in $[i',j']$ that vertices in $[i,j]$ see, then the sequence of these first vertices (viewed as integers) is a monotonically non-decreasing sequence once we omit all occurrences of $\nil$ from it.\footnote{A non-decreasing function (or sequence) is one that {\em never} decreases but can sometimes {\em not} increase.} We further demand that, between two equal vertices in this sequence, no $\nil$ occurs. Formally,

\begin{figure}[t]
\centering
\fbox{\includegraphics[scale=0.8]{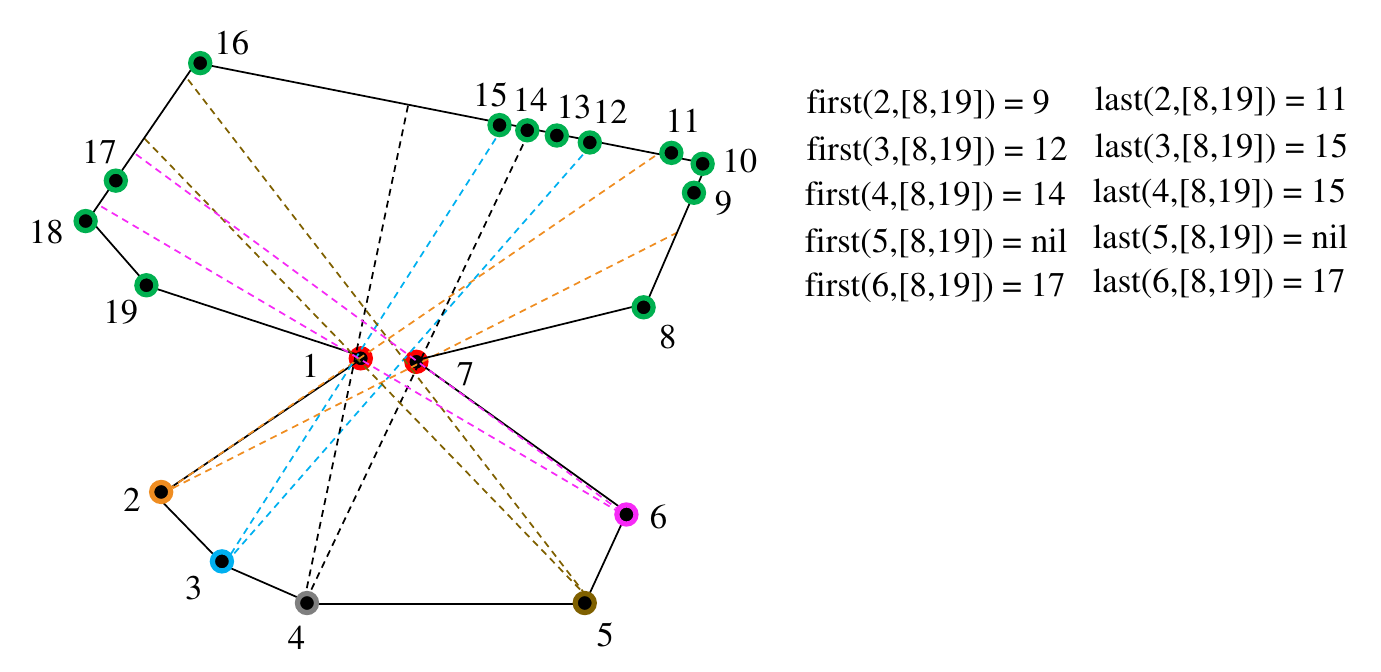}}
\caption{The way $[2,6]$ views $[8,19]$ is non-decreasing with respect to both $\first$ and $\last$.}\label{fig:view}
\end{figure}

\begin{definition}\label{def:first}
Let $P=(V,E)$ be a simple polygon. We say that the way a convex region $[i,j]$ of $P$ views a (not necessarily distinct) convex region $[i',j']$ of $P$ is {\em non-decreasing} (resp. {\em non-increasing}) {\em with respect to $\first$} if for all $t,\widehat{t}\in \{i,i+1,\ldots,j\}$ such that $t\leq \widehat{t}$, $\first(t,[i',j'])\neq\nil$ and $\first(\widehat{t},[i',j'])\neq\nil$, we have that
\begin{itemize}
\item $\first(t,[i',j'])\leq \first(\widehat{t},[i',j'])$ (resp.~$\first(t,[i',j'])\geq \first(\widehat{t},[i',j'])$), and
\item if $\first(t,[i',j'])=\first(\widehat{t},[i',j'])$, then for all $p\in\{t,\ldots,\widehat{t}\}$, $\first(p,[i',j'])=\first(t,[i',j'])$.\footnote{We remark that this condition cannot be replaced by ``for all $p\in\{t,\ldots,\widehat{t}\}$, $\first(p,[i',j'])\neq\nil$''. For example, in Fig.~\ref{fig:view}, neither $\first(4,[8,19])$ nor $\first(6,[8,19])$ is $\nil$, but $\first(5,[8,19])=\nil$.}
\end{itemize}
\end{definition}

Symmetrically, we address the orientation corresponding to the notation $\last$.

\begin{definition}\label{def:last}
Let $P=(V,E)$ be a simple polygon. We say that the way a convex region $[i,j]$ of $P$ views a (not necessarily distinct) convex region $[i',j']$ of $P$ is {\em non-decreasing} (resp. {\em non-increasing}) {\em with respect to $\last$} if for all $t,\widehat{t}\in \{i,i+1,\ldots,j\}$ such that $t\leq \widehat{t}$, $\last(t,[i',j'])\neq\nil$ and $\last(\widehat{t},[i',j'])\neq\nil$, we have that
\begin{itemize}
\item $\last(t,[i',j'])\leq \last(\widehat{t},[i',j'])$ (resp.~$\last(t,[i',j'])\geq \last(\widehat{t},[i',j'])$), and
\item if $\last(t,[i',j'])=\last(\widehat{t},[i',j'])$, then for all $p\in\{t,\ldots,\widehat{t}\}$, $\last(p,[i',j'])=\last(t,[i',j'])$.
\end{itemize}
\end{definition}

The main purpose of this section is to prove the following two lemmas. We believe that some arguments required to establish their proofs might be folklore. For the sake of completeness and self-containment, we present the full details. The first lemma asserts that the subsequence seen by a vertex within a convex region does not contain~``gaps''.

\begin{lem}\label{lem:seeAllBetFin}
Let $P=(V,E)$ be a simple polygon, $v\in V$, and $[i,j]$ be a convex region of $P$. Then, $v$ sees every vertex $t\in [i,j]$ such that $\first(v,[i,j])\leq t \leq \last(v,[i,j])$.\footnote{If $v$ does not see any vertex in $[i,j]$, the claim holds vacuously.}
\end{lem}
\begin{proof}
Suppose that $v$ sees some vertex in $[i,j]$, else the proof is trivial. Denote $\ell=\first(v,[i,j])$ and $h=\last(v,[i,j])$. We consider two cases. First, suppose that $v\notin [i,j]$. Define a polygon $Q=(V_Q,E_Q)$ by $V_Q=\{\ell,{\ell+1},\ldots,h\}\cup\{v\}$ and $E_Q=\{\{t,{t+1}\}: t\in\{\ell,\ldots,h-1\}\}\cup\{\{\ell,v\},\{h,v\}\}$. Clearly, $Q$ is simple. Since $[i,j]$ is a convex region of $P$, we have that $Q$ is a simple polygon such that the interior angle at $t$ in $Q$, for any $t\in\{\ell,\ell+1,\ldots,h\}$, is at most 180 degrees. Thus, the only vertex in $Q$ that can be a reflex vertex is $v$. Moreover, since $P$ contains both $\overline{\ell v}$ and $\overline{h v}$, we have that $Q$ is contained in $P$. By Observation \ref{obs:convexPolygon} and Proposition \ref{prop:reflex}, this means that for all $t\in\{\ell,\ldots,h\}$, $v$ sees $t$.

Second, suppose that $v\in[i,j]$. Define a polygon $Q=(V_Q,E_Q)$ by $V_Q=\{\ell,{\ell+1},\ldots,v\}$ and $E_Q=\{\{t,{t+1}\}: t\in\{\ell,\ldots,v-1\}\}\cup\{\{\ell,v\}\}$ and a polygon $Q'=(V_Q',E_Q')$ by $V_Q'=\{v,v+1,\ldots,h\}$ and $E_Q'=\{\{t,{t+1}\}: t\in\{v,\ldots,h-1\}\}\cup\{\{h,v\}\}$. Clearly, both polygons are simple and convex. Moreover, since $P$ contains both $\overline{\ell v}$ and $\overline{h v}$, we have that both $Q$ and $Q'$ are contained in $P$. By Observation \ref{obs:convexPolygon}, this means that for all $t\in\{\ell,\ldots,h\}$, $v$ sees $t$.
\end{proof}

The second lemma asserts that views are monotone. Intuitively, whenever we move along a convex region $[i,j]$ while viewing a convex region $[i',j']$ as described earlier, the first vertices (and last vertices) seen form a non-increasing or non-decreasing sequence.\footnote{We remark that we do not know whether it is possible that the first vertices would form a non-increasing (or non-decreasing) sequence and the last vertices would not. Our weaker claim suffices for our purposes.}

\begin{lem}\label{lem:mainStructF}
Let $P=(V,E)$ be a simple polygon, and let $[i,j]$ and $[i',j']$ be two (not necessarily distinct) maximal convex regions of $P$. Then, {\em (i)} the way in which $[i,j]$ views $[i',j']$ with respect to $\first$ is either non-decreasing or non-increasing, and {\em (ii)} the way in which $[i,j]$ views $[i',j']$ with respect to $\last$ is either non-decreasing or non-increasing.
\end{lem}

We only prove the first statement in Lemma \ref{lem:mainStructF}. (The proof of the second statement is symmetric.) To this end, we first analyze how a convex region sees itself, and afterwards we analyze how one convex region sees a different convex region. Having completed this analysis, we present the proof of the lemma.

\medskip
\noindent{\bf Interaction within the same region.} First, we analyze how a convex region sees itself.

\begin{lem}\label{lem:polyInConvex}
Let $P=(V,E)$ be a simple polygon. Let $[i,j]$ a convex region of $P$. Let $\ell,h\in [i,j]$ be two vertices that see each other, where $\ell\leq h$. For all $x,y\in\{\ell,\ell+1,\ldots,h\}$, $x\leq y$, the vertices $x$ and $y$ see each other.
\end{lem}

\begin{proof}
Define the polygon $Q=(V_Q,E_Q)$ by $V_Q=\{\ell,{\ell+1},\ldots,h\}$ and $E_Q=\{\{t,{t+1}\}: t\in\{\ell,\ldots,h-1\}\}\cup\{\{\ell,h\}\}$. Since $[i,j]$ is a convex region of $P$ and the line segment $\overline{\ell h}$ is contained in $P$, we have that $Q$ is a convex polygon that is contained in $P$. By Observation \ref{obs:convexPolygon}, this means that any two vertices of $Q$ see each other.
\end{proof}

We utilize Lemma \ref{lem:polyInConvex} in order to prove the following result.

\begin{lem}\label{lem:seeSame}
Let $P=(V,E)$ be a simple polygon. Let $[i,j]$ be a convex region of $P$. Let ${\ell}$ and $h$ be two vertices in $[i,j]$ such that $\ell\leq h$, $x=\first(\ell,[i,j])\neq\nil$, and $y=\first(h,[i,j])\neq\nil$. Then, for all $t\in\{\ell,\ell+1,\ldots,h\}$, $\min\{x,y\}\leq\first(t,[i,j])\leq \max\{x,y\}$.
\end{lem}

\begin{proof}
Suppose that $\ell<h-1$, else the proof is complete.
Let $t\in\{\ell+1,\ldots,h-1\}$. Suppose, by way of contradiction, that either $\first(t,[i,j])<\min\{x,y\}$ or $\max\{x,y\}<\first(t,[i,j])$. First, assume that $\first(t,[i,j])<\min\{x,y\}$. Because every vertex sees itself, we have that $\min\{x,y\}\leq\ell$. Thus, $\first(t,[i,j])<\ell<t$. By Lemma \ref{lem:polyInConvex}, this implies that $\ell$ sees $\first(t,[i,j])$. However, this is contradiction because $x=\first(\ell,[i,j])$ while $\first(t,[i,j])<x$.
Second, assume that $\max\{x,y\}<\first(t,[i,j])$. Since every vertex sees itself, we have that $\first(t,[i,j])\leq t$, and hence $\max\{x,y\}<t$. In particular, $y<t<h$. By Lemma \ref{lem:polyInConvex}, this implies that $t$ sees $y$. However, this is contradiction because $y<\first(t,[i,j])$.
\end{proof}

\begin{figure}[t]
\centering
\fbox{\includegraphics[scale=0.8]{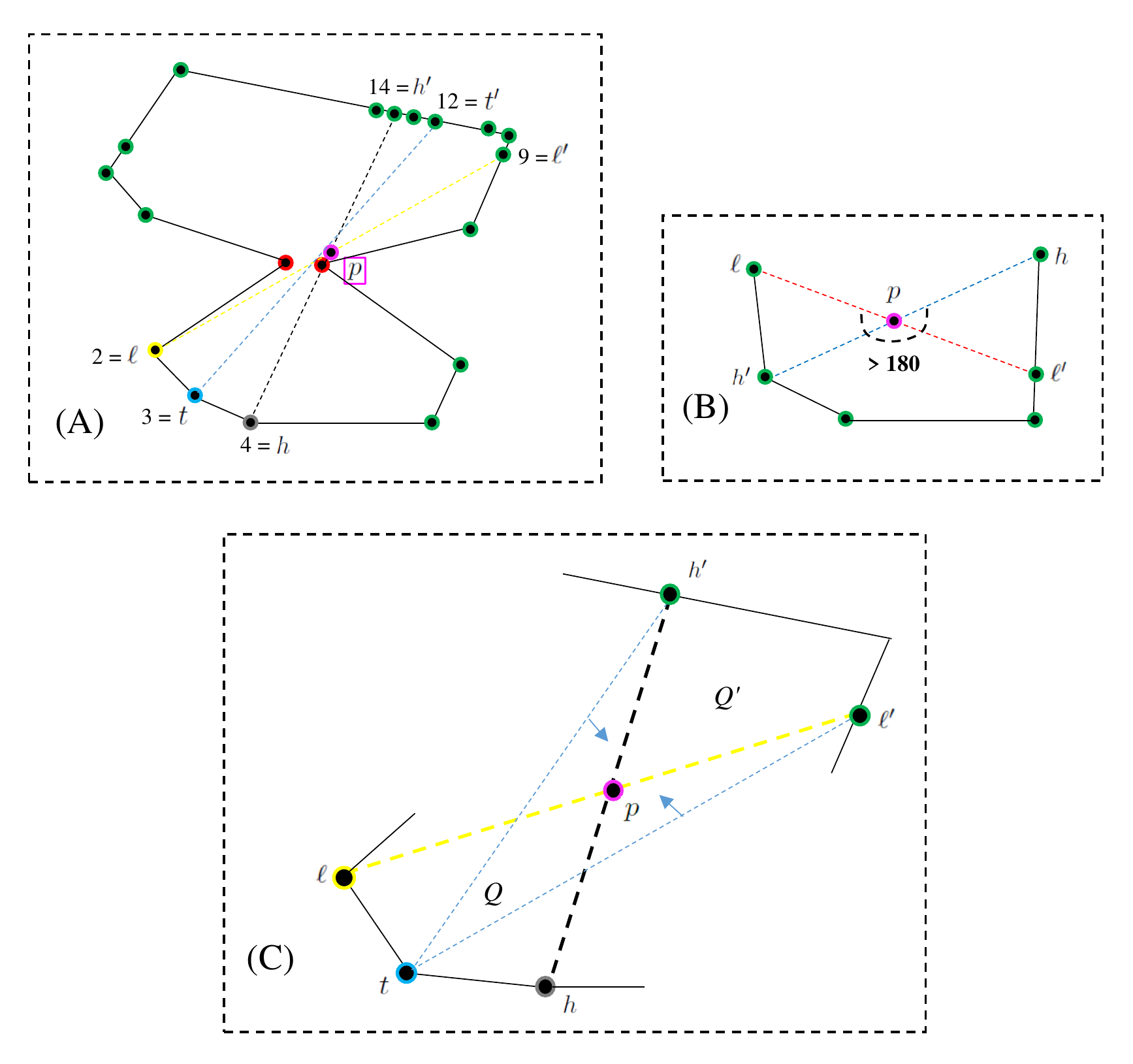}}
\caption{(A) The vertices $\ell,\ell',h,h',t,t'$ and $p$ in the proof of Lemma \ref{lem:seeSeveralIntersect}. The polygon is the same as the one in Fig.~\ref{fig:view}. (B) A contradiction in the proof of Lemma~\ref{lem:seeSeveralIntersect}: the vertices $\ell'$ and $h'$ belong to the same convex region as the vertices $\ell$ and $h$. (C) The line segment $\overline{tt'}$ must intersect both $\overline{\ell\ell'}$ and $\overline{hh'}$.}\label{fig:uniqueIntersection}
\end{figure}

\medskip
\noindent{\bf Interaction between two distinct regions.} Second, we analyze how one convex region sees a different convex region. For this purpose, we first argue that certain line segments intersect. Then, we consider the case where they intersect in a single point, and the case where they intersect in more than a single point.

\begin{lem}\label{lem:seeIntersect}
Let $P=(V,E)$ be a simple polygon. Let $[i,j]$ and $[i',j']$ be distinct maximal convex regions of $P$. Let ${\ell}$ and $h$ be vertices in $[i,j]$ such that $\ell\leq h$, $\ell'=\first(\ell,[i',j'])\neq\nil$ and ${h'}=\first(h,[i',j'])\neq\nil$. Then, the line segments $\overline{\ell {\ell'}}$ and~$\overline{h {h'}}$~intersect.
\end{lem}

\begin{proof}
Suppose, by way of contradiction, that $\overline{\ell {\ell'}}$ and $\overline{h {h'}}$ do not intersect. Then, $\ell\neq h$ and $\ell'\neq h'$. 
Define a polygon $Q=(V_Q,E_Q)$ by $V_Q=\{\ell,{\ell+1},\ldots,h\}\cup\{\min({\ell'},h'),\ldots,\max(\ell',{h'})\}$ and $E_Q=\{\{t,{t+1}\}: t\in\{\ell,\ldots,h-1\}\}\cup\{\{{t'},{t'+1}\}: t'\in\{\min(\ell',h'),\ldots,\max(\ell',h')-1\}\}\cup\{\{\ell,\ell'\},\{h,h'\}\}\}$. For any vertex $v\in V_Q\setminus\{\ell,h,{\ell'},{h'}\}$, the interior angle at $v$ is the same in $Q$ and $P$. Moreover, for each any $v\in \{\ell,h,{\ell'},{h'}\}$, because $\overline{\ell {\ell'}}$ and $\overline{h {h'}}$ are contained in $P$, the interior angle at $v$ in $Q$ is at most the interior angle at $v$ in $P$. Thus, since $[i,j]$ and $[i',j']$ are convex region of $P$, we have that any interior angle of $Q$ is at most 180 degrees. Moreover, because the line segments $\overline{\ell {\ell'}}$ and $\overline{h {h'}}$ do not intersect, we have that $Q$ is simple. Thus, $Q$ is a convex polygon contained in $P$.
By Observation \ref{obs:convexPolygon}, $h$ sees ${\ell'}$ in $Q$, and $\ell$ sees ${h'}$ in $Q$. In turn, this implies that $h$ sees ${\ell'}$ in $P$, and $\ell$ sees ${h'}$ in $P$. 
If $\ell'<h'$, then ${\ell'}<{h'}=\first(h,[i',j'])$, which is a contradiction. Hence, $\ell'>h'$. However, then ${h'}<{\ell'}=\first(\ell,[i',j'])$, which is a contradiction.
\end{proof}

Now, we analyze the case where the intersection consists of a single point.

\begin{lem}\label{lem:seeUniqueIntersect}
Let $P=(V,E)$ be a simple polygon. Let $[i,j]$ and $[i',j']$ be two distinct maximal convex regions of $P$. Let ${\ell}$ and $h$ be two vertices in $[i,j]$ such that $\ell\leq h$, ${\ell'}=\first(\ell,[i',j'])\neq\nil$, ${h'}=\first(h,[i',j'])\neq\nil$ and the line segments $\overline{\ell {\ell'}}$ and $\overline{h {h'}}$ intersect at a single point. Then, for all $t\in\{\ell,\ell+1,\ldots,h\}$, either $\first(t,[i',j'])=\nil$ or $\min\{{\ell'},{h'}\}\leq\first(t,[i',j'])\leq \max\{{\ell'},{h'}\}$.
\end{lem}

\begin{proof}
Suppose that $\ell<h-1$, else the proof is complete.
Let $p$ denote the unique point where $\overline{\ell {\ell'}}$ and $\overline{h {h'}}$ intersect.
Define two polygons as follows (see Fig.~\ref{fig:uniqueIntersection}(A)).
\begin{itemize}
\item The first polygon $Q=(V_Q,E_Q)$ is given by $V_Q=\{\ell,\ldots,h\}\cup\{p\}$ and $E_Q=\{\{t,{t+1}\}: t\in\{\ell,\ldots,h-1\}\}\cup\{\{h,p\},\{p,\ell\}\}$.
\item The second polygon $Q'=(V_{Q'},E_{Q'})$ is given by $V_{Q'}=\{\min({\ell'},h'),\ldots,\max(\ell',{h'})\}\cup\{p\}$ and $E_{Q'}=\{\{{t'},{t'+1}\}: t'\in\{\min({\ell'},h'),\ldots,\max(\ell',{h'})-1\}\}\cup\{\{{h'},p\},\{p,{\ell'}\}\}$.
\end{itemize}
We claim that $Q$ and $Q'$ are convex polygons contained in $P$.
We prove this claim only for $Q$ since the proof for $Q'$ is symmetric. First, since $[i,j]$ is a convex region of $P$, and the line segments $\overline{\ell p}$ and $\overline{h p}$ intersect only at $p$ and are contained in $P$, we have that $Q$ is a simple polygon that is contained in $P$. Moreover, every interior angle at $t$ in $Q$, for all $t\in\{\ell,\ell+1,\ldots,h\}$, is at most the interior angle at $t$ in $P$, and hence it is at most 180 degrees. Now, consider the interior angle at $p$ in $Q$. If this angle were larger than 180 degrees, then ${\ell'}$ and ${h'}$ would have belonged to $[i,j]$ (see Fig.~\ref{fig:uniqueIntersection}(B)), which yields a contradiction since $[i,j]$ and $[i',j']$ are distinct maximal convex regions of $P$. Thus, $Q$ is convex.

Towards the proof that for all $t\in\{\ell+1,\ldots,h-1\}$, either $\first(t,[i',j'])=\nil$ or $\min\{{\ell'},{h'}\}\leq\first(t,[i',j'])\leq \max\{{\ell'},{h'}\}$, choose some $t\in\{\ell+1,\ldots,h-1\}$, and denote $t'=\first(t,[i',j'])$. If $t'=\nil$, then we are done. Thus, suppose that ${t'}\neq\nil$. By Lemma \ref{lem:seeIntersect}, the line segment $\overline{tt'}$ intersects both $\overline{\ell\ell'}$ and $\overline{hh'}$. Since the polygons $Q$ and $Q'$ are convex, and because $t$ belongs to $Q$, this is only possible if $t'$ belongs to the boundary of $Q'$ that coincides with the convex region $[i',j']$ of $P$ (see Fig.~\ref{fig:uniqueIntersection}(C)). From this, we conclude that $\min\{{\ell'},{h'}\}\leq t'\leq \max\{{\ell'},{h'}\}$.
\end{proof}

Secondly, we analyze the case where the intersection consists of more than a single point.

\begin{lem}\label{lem:seeSeveralIntersect}
Let $P=(V,E)$ be a simple polygon. Let $[i,j]$ and $[i',j']$ be two distinct maximal convex regions of $P$. Let ${\ell}$ and $h$ be two vertices in $[i,j]$ such that $\ell\leq h$, ${\ell'}=\first(\ell,[i',j'])\neq\nil$, ${h'}=\first(h,[i',j'])\neq\nil$ and the line segments $\overline{\ell {\ell'}}$ and $\overline{h {h'}}$ intersect at more than one point. Then, for all $t\in\{\ell,\ell+1,\ldots,h\}$, $\min\{{\ell'},{h'}\}=\first(t,[i',j'])=\max\{{\ell'},{h'}\}$.
\end{lem}

\begin{proof}
Since $[i,j]$ is a convex region of $P$, and because $\overline{\ell {\ell'}}$ and $\overline{h {h'}}$ intersect at more than one point, we have that the interior angle at $t$ in $P$, for all $t\in\{\ell+1,\ldots,h-1\}$, is exactly 180 degrees (see Fig.~\ref{fig:nonUniqueIntersection}(A)).
Then, $\ell$ sees ${h'}$ and $h$ sees ${\ell'}$, which implies that $\ell'=h'$. Thus, one of the line segments $\overline{\ell {h'}}$ and $\overline{h {h'}}$ is a subsegment of the other. Without loss of generality, suppose that $\overline{h {h'}}$ is a subsegment of $\overline{\ell {h'}}$, and that $\overline{\ell {h'}}$ and $\overline{h {h'}}$ are parallel to the $x$ axis. Note that this means that the interior angle at $h$ in $P$ is also 180 degrees.

Suppose that $\ell<h-1$, else the proof is complete. Let $t\in\{\ell+1,\ldots,h-1\}$, and denote $t'=\first(t,[i',j'])$. We need to prove that $t'=h'$. Suppose, by way of contradiction, that $t'\neq h'$. Because $t$ sees $h'$, this means that $t'<h'$. Observe that $t$ sees $h'$, and $t$ does not see any vertex in $[i',j']$ whose $y$-coordinate is lower than the $y$-coordinate of $h'$.  Thus, the $y$-coordinate of ${t'}$ is larger than the one of $t$. Then, the polygon defined by $\overline{{t'} t}, \overline{t h},\overline{h {h'}}$ and $\overline{q(q+1)}$ for all $q\in\{t',\ldots,h'-1\}$ is convex and contained in $P$ (see Fig.~\ref{fig:nonUniqueIntersection}(B)). However, by Observation \ref{obs:convexPolygon}, this means that $h$ sees ${t'}$, and hence $\first(h,[i',j'])$ cannot be equal to ${h'}$ (because $t'<h'$). We have thus reached a contradiction, which concludes the proof.
\end{proof}

\begin{figure}[t]
\centering
\fbox{\includegraphics[scale=0.8]{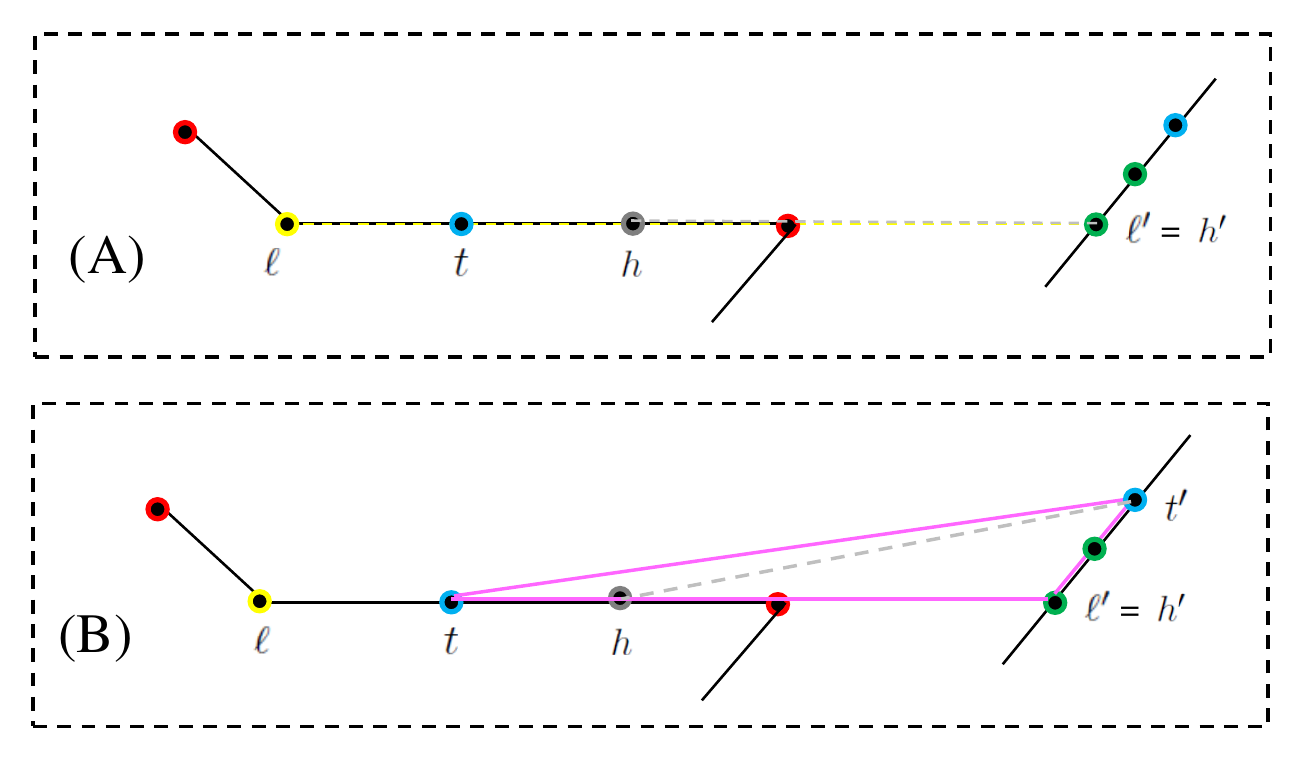}}
\caption{(A) The vertices $\ell,\ell',h,h'$ and $t$ in the proof of Lemma~\ref{lem:seeSeveralIntersect}. (B) The polygon defined in the proof of Lemma~\ref{lem:seeSeveralIntersect}.}\label{fig:nonUniqueIntersection}
\end{figure}

From Lemmas \ref{lem:seeIntersect}, \ref{lem:seeUniqueIntersect} and \ref{lem:seeSeveralIntersect}, we derive the following result.

\begin{lem}\label{lem:seeDistinct}
Let $P=(V,E)$ be a simple polygon. Let $[i,j]$ and $[i',j']$ be two distinct maximal convex regions of $P$. Let ${\ell}$ and $h$ be two vertices in $[i,j]$ such that $\ell\leq h$, ${\ell'}=\first(\ell,[i',j'])\neq\nil$, ${h'}=\first(h,[i',j'])\neq\nil$. Then, for all $t\in\{\ell,\ell+1,\ldots,h\}$, either $\first(t,[i',j'])=\nil$ or $\min\{{\ell'},{h'}\}\leq\first(t,[i',j'])\leq \max\{{\ell'},{h'}\}$.
\end{lem}

\medskip
\noindent{\bf Proof of the first statement of Lemma \ref{lem:mainStructF}.} 
Suppose, by way of contradiction, that the way in which $[i,j]$ views $[i',j]$ with respect to $\first$ is neither non-decreasing nor non-increasing. Then, there exist $x,y,z\in\{i,i+1,\ldots,j\}$ such that $x<y<z$, $\first(x,[i',j'])\neq\nil$, $\first(z,[i',j'])\neq\nil$, and
\begin{enumerate}
\item\label{cond1} $\max\{\first(x,[i',j']),\first(z,[i',j'])\}<\first(y,[i',j'])$, or
\item\label{cond2} $\min\{\first(x,[i',j']),\first(z,[i',j'])\}>\first(y,[i',j'])$, or
\item\label{cond3} $\first(x,[i',j']) = \first(z,[i',j'])$ and $\first(y,[i',j'])=\nil$.
\end{enumerate}

If $\first(x,[i',j']) = \first(z,[i',j'])$, then by Lemma \ref{lem:seeAllBetFin}, $\first(x,[i',j'])$ sees $t$ for all $t\in\{x,x+1,\ldots,z\}$. Thus, the third condition cannot be satisfied.
If $[i,j]\neq [i',j']$, then Lemma \ref{lem:seeDistinct} implies that neither of the first two conditions can be satisfied. Otherwise, if $[i,j]=[i',j']$, then Lemma \ref{lem:seeSame} implies that neither of the first two conditions can be satisfied. Thus, we necessarily reach a contradiction.\qed

%


\subsection{Turing Reduction to Structured Art Gallery}\label{sec:structured}
An intermediate step in our reduction from {\sc Art Gallery} to {\sc Monotone 2-CSP} addresses an annotated version of {\sc Art Gallery}, called {\sc Structured Art Gallery}. Intuitively, in {\sc Structured Art Gallery} each convex region ``announces'' how many guards it should contain, and how many guards are to be used to see it completely. In addition, each convex region announces by which {\em unknown guard} (identified as ``the $i^{th}$ guard to be placed on region $C$'' for some $i$ and $C$) its prefix should be guarded, by which unknown guard a region after this prefix should be guarded, and so on.
In what follows, we formally define the {\sc Structured Art Gallery} problem; then, we present our reduction from {\sc Art Gallery} to {\sc Structured Art Gallery}, and afterwards argue that this reduction is correct.
For a polygon $P$, let ${\cal C}(P)$ be the set of maximal convex regions of $P$. Note that $|{\cal C}(P)|\leq r$. 

\medskip
\noindent{\bf Problem Definition.} The input of {\sc Structured Art Gallery} consists of a simple polygon $P=(V,E)$, a non-negative integer $k<r$, and the following functions (see Fig.~\ref{fig:turing1}).
\begin{itemize}
\item $\ig: {\cal C}(P)\cup\reflex(P)\rightarrow \{0,\ldots,k\}$, where $\sum_{x\in {\cal C}(P)\cup\reflex(P)}\ig(x)\leq k$. Intuitively, for a convex region or reflex vertex $x$, $\ig$ assigns the number of guards to be placed in $x$.

\item  $\og: {\cal C}(P)\cup\reflex(P)\rightarrow \{1,\ldots,k\}$, where for all $x\in\reflex(P)$, $\og(x)=1$.  Intuitively, for a convex region or reflex vertex $x$, $\og$ assigns the number of guards required to see $x$.

\item For each $x\in {\cal C}(P)\cup\reflex(P)$, $\how_x: \{1,\ldots,\og(x)\}\rightarrow ({\cal C}(P)\cup\reflex(P))\times\{1,\ldots,k\}$, where for each $(y,i)$ in the image of $\how_x$, $i\leq \ig(y)$. Intuitively, for any $j\in\{1,\ldots,\og(x)\}$, $\how_x(j)=(y,i)$ indicates that the $j^{th}$ guard required to see $x$ is the $i^{th}$ guard placed~in~$y$.
\end{itemize}

\begin{figure}[t]
\centering
\fbox{\includegraphics[scale=0.8]{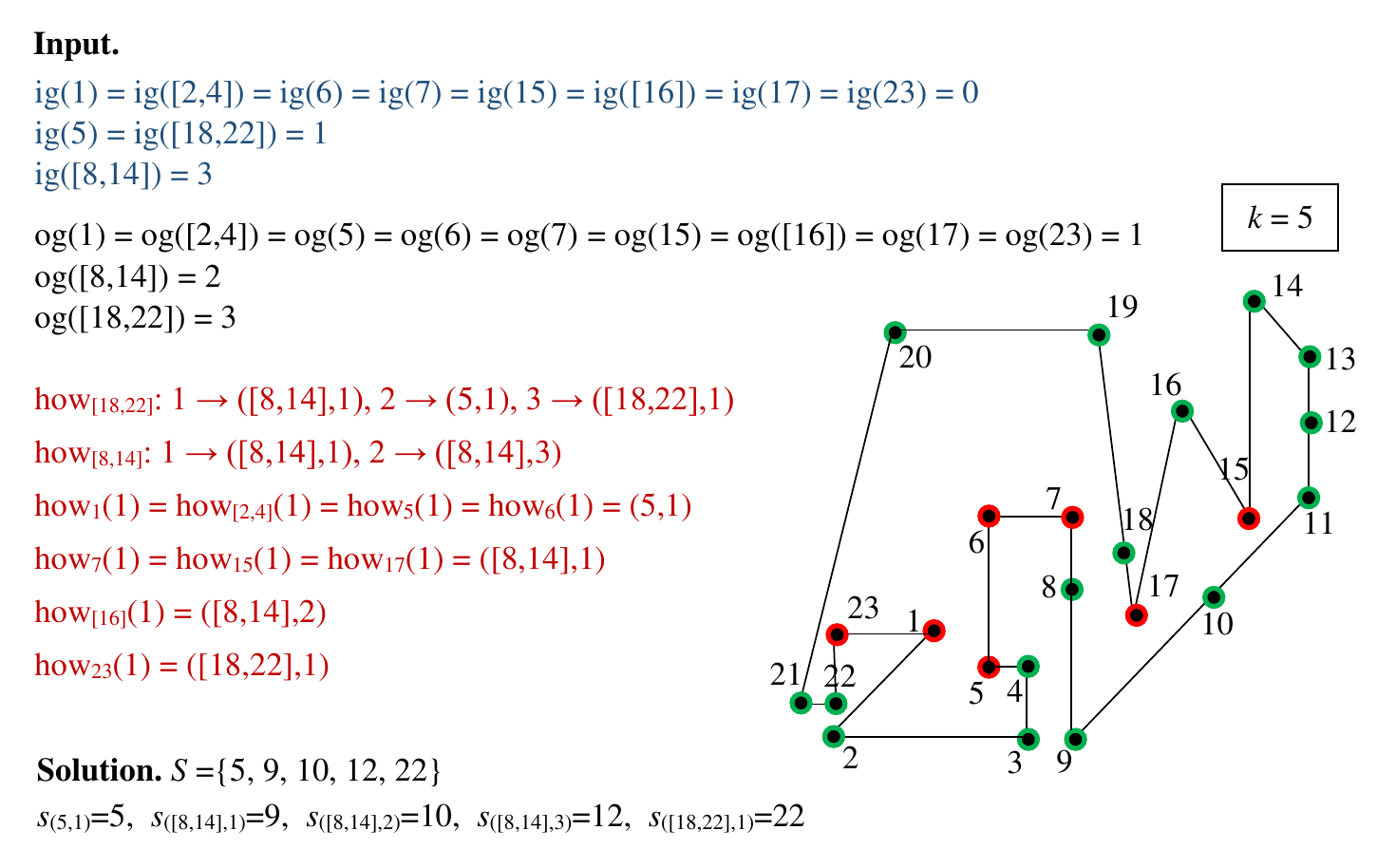}}
\caption{An input and a solution for the {\sc Structured Art Gallery} problem.}\label{fig:turing1}
\end{figure}

The objective of {\sc Structured Art Gallery} is to determine whether there exists a set $S\subseteq V$ of size at most $k$ such that the following conditions hold:
\begin{enumerate}
\item\label{condition:structured1} For each $x\in {\cal C}(P)\cup\reflex(P)$, $|S\cap x|=\ig(x)$.\footnote{If $x\in\reflex(P)$, by $S\cap x$ we mean $S\cap\{x\}$.} Accordingly, for each $x\in {\cal C}(P)\cup\reflex(P)$ and $i\in\{1,\ldots,\ig(x)\}$, let $s_{(x,i)}$ denote the $i^{th}$ largest vertex in $S\cap x$ (see Fig.~\ref{fig:turing1}).
\item\label{condition:structured2}  For each $x\in\reflex(P)$, $s_{\how_x(1)}$ sees $x$.
\item\label{condition:structured3}  For each $C\in {\cal C}(P)$, the following conditions hold:
	\begin{enumerate}
	\item\label{condition:structured3a} $\first(s_{\how_C(1)},C)$ is the smallest vertex in $C$.
	\item\label{condition:mid} For every $t\in \{1,\ldots,\og(C)-1\}$, denote $i=\last(s_{\how_C(t)},C)$, $j=\first(s_{\how_C(t+1)},C)$ and $q=\last(s_{\how_C(t+1)},C)$. Then, {\em (i)} $i\geq j-1$, and {\em (ii)} $i\leq q-1$. (See Fig.~\ref{fig:turing2}.)
	\item\label{condition:structured3c} $\last(s_{\how_C(\og(C))},C)$ is the largest vertex in $C$.
	\end{enumerate}
\end{enumerate}
Informally, Condition \ref{condition:mid} states that {\em (i)} the last vertex in $C$ seen by its $t^{th}$ guard should be at least as large as the predecessor of the first vertex in $C$ seen by its $(t+1)^{th}$ guard, and {\em (ii)} the last vertex in $C$ seen by its $t^{th}$ guard should be smaller than the last vertex in $C$ seen by its $(t+1)^{th}$ guard. The first condition ensures that no unseen ``gaps'' are created within $C$, while the second condition ensures that as the index $t$ grows larger, the last vertex seen by the $t^{th}$ guard grows larger as well. (The second condition will be part of our transition towards the interpretation of the objective of {\sc Art Gallery} by {\em binary} constraints.)

\begin{figure}[t]
\centering
\fbox{\includegraphics[scale=0.8]{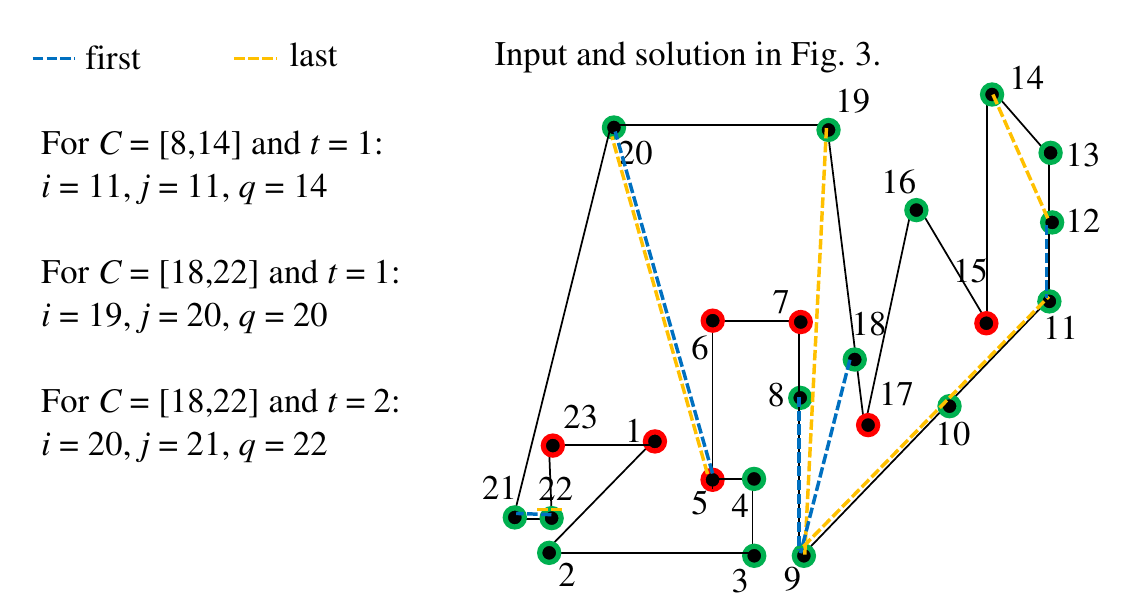}}
\caption{Condition \ref{condition:mid} satisfied by a solution for {\sc Structured Art Gallery}.}\label{fig:turing2}
\end{figure}

\medskip
\noindent{\bf Turing Reduction.} Given an instance $(P,k)$ of {\sc Art Gallery}, in case $r\leq k$, output \Yes.\footnote{To comply with the formal definition of a Turing reduction, by \Yes\ we mean a set with a single trivial \Yes-instance of {\sc Structured Art Gallery}.}
Otherwise, the output of the reduction, $\red(P,k)$, is the set of all instances $(P,k,\ig,\og,\{\how_x\}|_{x\in {\cal C}(P)\cup\reflex(P)})$ of {\sc Structured Art Gallery}.

Observe that $|{\cal C}(P)\cup\reflex(P)|\leq 2r$, and therefore the number of possible functions $\ig$ is upper bounded by $(k+1)^{2r}$, the number of possible functions $\og$ is upper bounded by $k^{2r}$, and for each $x\in {\cal C}(P)\cup\reflex(P)$, the number of possible functions $\how_x$ is upper bounded by $(2rk)^k$. Hence, the number of instances produced is upper bounded by $(k+1)^{2r}\cdot k^{2r}\cdot ((2rk)^k)^{2r}$. When $k\leq r$, this number is upper bounded by $r^{\OO(r^2)}$. Moreover, the instances in $\red(P,k)$ can be enumerated with polynomial delay. Thus, 

\begin{obs}\label{obs:redToStructTime}
Let $(P,k)$ be an instance of {\sc Art Gallery}. Then, $|\red(P,k)|=r^{\OO(r^2)}$, and $\red(P,k)$ is computable in time $r^{\OO(r^2)}n^{\OO(1)}$.
\end{obs}

\medskip
\noindent{\bf Correctness.} Our proof of correctness crucially relies on Lemma \ref{lem:seeAllBetFin} and Proposition \ref{prop:reflex}. 

\begin{lem}\label{lem:correct-Turing-reduction-vvag}
An instance $(P,k)$ is a \Yes-instance of {\sc Art Gallery} if and only if there is a \Yes-instance of {\sc Structured Art Gallery} in $\red(P,k)$.
\end{lem}

\begin{proof}
{\bf Forward Direction.} Suppose that $(P,k)$ is a \Yes-instance of {\sc Art Gallery} and that $r>k$. Accordingly, let $S\subseteq V$ be a solution to $(P,k)$. We first define the function $\ig: {\cal C}(P)\cup\reflex(P)\rightarrow \{0,\ldots,k\}$ as follows. For each $x\in {\cal C}(P)\cup\reflex(P)$, let $\ig(x)=|S\cap x|$. Because $|S|\leq k$ (since $S$ is a solution to $(P,k)$), we have that $\sum_{x\in {\cal C}(P)\cup\reflex(P)}\ig(x)\leq k$. For each $x\in {\cal C}(P)\cup\reflex(P)$, we order the vertices in $S\cap x$ from smallest to largest, and denote them accordingly by $s_{(x,1)},s_{(x,2)},\ldots,s_{(x,\ig(x))}$.

Now, we define the functions $\og: {\cal C}(P)\cup\reflex(P)\rightarrow \{1,\ldots,k\}$ and $\how_x: \{1,\ldots,\og(x)\}\rightarrow ({\cal C}(P)\cup\reflex(P))\times\{1,\ldots,k\}$ for all $x\in {\cal C}(P)\cup\reflex(P)$.
For each reflex vertex $x\in\reflex(P)$, define $\og(x)=1$, and $\how_x(1)={(y,i)}$ for some vertex $s_{(y,i)}\in S$ that sees $x$. The existence of such a vertex $s_{(y,i)}$ follows from the assertion that $S$ is a solution to $(P,k)$.
For each convex region $C\in{\cal C}(P)$, define $\og(C)$ and $\how_C$ as follows. Let $W$ denote the set of vertices in $S$ that see at least one vertex in $C$. Since $W$ sees $C$, there exists a vertex in $W$ that sees the smallest vertex in $C$. Pick such a vertex arbitrarily and denote it by $w_1$. Now, if $w_1$ does not see the largest vertex in $C$, then there exists a vertex in $W$ that sees the smallest vertex in $C$ that is larger than the largest vertex seen by $w_1$. We pick such a vertex arbitrarily, and denote it by $w_2$. Next, if $w_2$ does not see the largest vertex in $C$, then there exists a vertex in $W$ that sees the smallest vertex in $C$ that is larger than the largest vertex seen by $w_2$. We pick such a vertex arbitrarily, and denote it by $w_3$. Similarly, we define $w_4,w_5,\ldots,w_p$, for the appropriate $p\in\{1,\ldots,k\}$ (see Fig.~\ref{fig:selection}). Here, the supposition that $p\leq k$ follows from Lemma \ref{lem:seeAllBetFin}, which implies that $w_i\neq w_j$ for all distinct $i,j\in\{1,\ldots,p\}$.
We define $\og(C)=p$, and for all $t\in\{1,\ldots,\og(C)\}$, we define $\how_C(t)=(y,i)$ for the pair $(y,i)\in ({\cal C}(P)\cup\reflex(P))\times\{1,\ldots,k\}$ that satisfies $w_t=s_{(y,i)}$.

\begin{figure}[t]
\centering
\fbox{\includegraphics[scale=0.75]{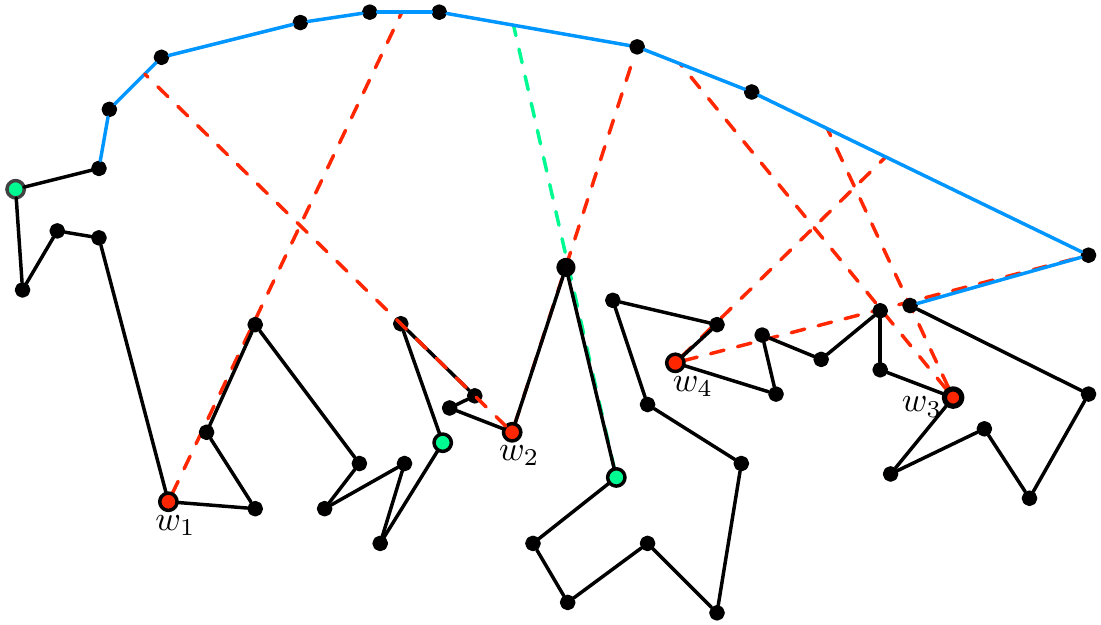}}
\caption{Example of a possible selection of $w_1,w_2,\ldots,w_p$. Solution vertices are colored green and red, and $C$ is colored blue.}\label{fig:selection}
\end{figure}

Our definitions directly ensure that for each $C\in {\cal C}(P)$, the following conditions hold:
	\begin{enumerate}
	\item $\first(s_{\how_C(1)},C)$ is the smallest vertex in $C$.
	\item For every $t\in \{1,\ldots,\og(C)-1\}$, denote $i=\last(s_{\how_C(t)},C)$, $j=\first(s_{\how_C(t+1)},C)$ and $q=\last(s_{\how_C(t+1)},C)$. Then, {\em (i)} $i\geq j-1$, and {\em (ii)} $i\leq q-1$.
	\item $\last(s_{\how_C(\og(C))},C)$ is the largest vertex in $C$.
	\end{enumerate}
By the arguments above, $I=(P,k,\ig,\og,\{\how_x\}|_{x\in {\cal C}(P)\cup\reflex(P)})$ is an instance of {\sc Structured Art Gallery}, and $S$ is a solution to $I$. Since $I\in\red(P,k)$, the proof of the forward direction is complete.

\medskip
\noindent{\bf Reverse Direction.} If $k\geq r$, then we output \Yes\ (or rather a trivial \Yes-instance), and by Proposition \ref{prop:reflex}, indeed the input is a \Yes-instance\ as well. Next, suppose that $k<r$, and there is a \Yes-instance $I=(P,k,\ig,\og,\{\how_x\}|_{x\in {\cal C}(P)\cup\reflex(P)})$ in $\red(P,k)$. Accordingly, let $S\subseteq V$ be a solution to $I$. Then, $|S|\leq k$. Thus, to prove that $(P,k)$ is a \Yes-instance of {\sc Art Gallery}, it suffices to show that $S$ sees $V$.
For each $x\in\reflex(P)$, $s_{\how_x(1)}$ sees $x$, and therefore $S$ sees $\reflex(P)$.

Now, we show that $S$ sees $\con(P)$. To this end, we choose a convex region $[i,j]\in {\cal C}(P)$, and show that $S$ sees $[i,j]$. Specifically, for each $p\in\{i,\ldots,j\}$, we prove that there is $t\in\{1,\ldots,\og([i,j])\}$ such that $s_{\how_{[i,j]}(t)}$ (which is a vertex in $S$) sees $p$. The proof is by induction on $p$.
In the basis, where $p=i$, correctness follows from the assertion that $\first(s_{\how_{[i,j]}(1)},{[i,j]})$ is the smallest vertex in $[i,j]$.
Now, we suppose that the claim is correct for $p$, and prove it for ${p+1}$. By the inductive hypothesis, there is $t\in\{1,\ldots,\og([i,j])\}$ such that $s_{\how_{[i,j]}(t)}$ sees $p$. If $s_{\how_{[i,j]}(t)}$ sees ${p+1}$, then we are done. Thus, we now suppose that $s_{\how_{[i,j]}(t)}$ does not see ${p+1}$. Then, $\last(s_{\how_{[i,j]}(t)},[i,j])=p$. We have two~cases:
\begin{itemize}
\item First, consider the case where $t<\og([i,j])$. Then, because $S$ is a solution to $I$, the vertex $p=\last(s_{\how_{[i,j]}(t)},[i,j])$ is larger or equal to $d-1$ for $d=\first(s_{\how_{[i,j]}(t+1)},[i,j])$. This means that $\first(s_{\how_{[i,j]}(t+1)},[i,j])\leq {p+1}$. Moreover, $p$ is smaller than the vertex $\last(s_{\how_{[i,j]}(t+1)},[i,j])$. Thus, ${p+1}\leq\last(s_{\how_{[i,j]}(t+1)},[i,j])$. Then, $\first(s_{\how_{[i,j]}(t+1)},[i,j])$ $\leq {p+1}\leq\last(s_{\how_{[i,j]}(t+1)},[i,j])$. By Lemma \ref{lem:seeAllBetFin}, this means that $s_{\how_{[i,j]}(t+1)}$ sees ${p+1}$.

\item Second, consider the case where $t=\og([i,j])$. In this case, because $S$ is a solution to $I$, we have that $\last(s_{\how_{[i,j]}(\og([i,j]))},[i,j])$ is the largest vertex in $[i,j]$. Thus, ${p+1}\leq\last(s_{\how_{[i,j]}(\og([i,j]))},[i,j])$, which is a contradiction.
\end{itemize}
This completes the proof.
\end{proof}

\subsection{Karp Reduction to Monotone 2-CSP}\label{sec:reduction}
We proceed to the second part of our proof, a reduction from {\sc Structured Art Gallery} to {\sc Monotone 2-CSP}.\footnote{CSP is an abbreviation of Constraint Satisfaction Problem, and 2 is the maximum arity of a constraint.} 

\medskip
\noindent{\bf Problem Definition.} The input of {\sc Monotone 2-CSP} consists of a set $X$ of {\em variables}, denoted by $X=\{x_1,x_2,\ldots,x_{|X|}\}$, a set $C$ of {\em constraints}, and $N\in\mathbb{N}$ given in unary. Each constraint $c\in C$ has the form $[x_i \si f(x_j)]$ where $i,j\in\{1,\ldots,|X|\}$, $\si\in\{\geq,\leq\}$ and $f: \{0,\ldots,N\}\rightarrow \{0,\ldots,N\}$ is a monotone function. An assignment $\alpha: X\rightarrow \{0,\ldots,N\}$ {\em satisfies} a constraint $c=[x_i \si f(x_j)]\in C$ if $[\alpha(x_i) \si f(\alpha(x_j))]$ is true. The objective of {\sc Monotone 2-CSP} is to decide if there exists an assignment $\alpha: X\rightarrow\{0,\ldots,N\}$ that satisfies all the constraints in $C$ (see Fig.~\ref{fig:csp}).

\begin{figure}[t]
\centering
\fbox{\includegraphics[scale=0.8]{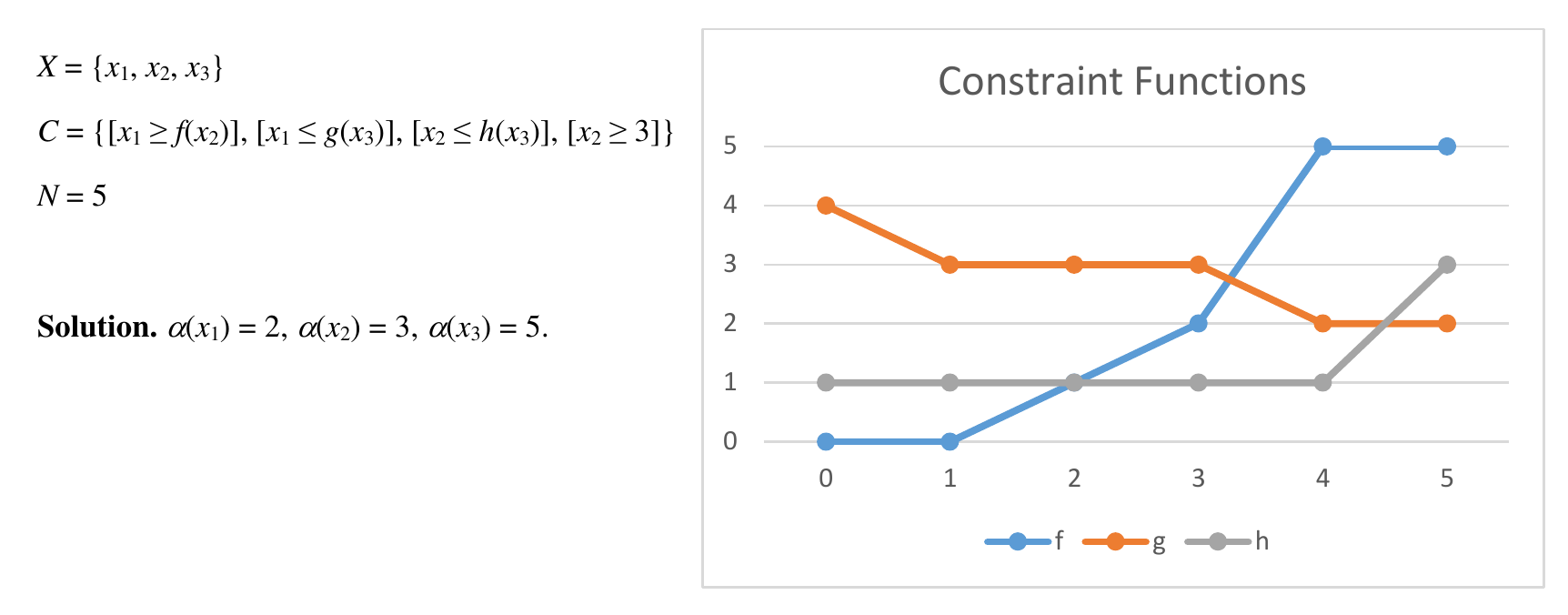}}
\caption{An input for {\sc Monotone 2-CSP} that has a unique solution.}\label{fig:csp}
\end{figure}

If the function $f$ of a constraint $c=[x_i \si f(x_j)]$ is constantly $\beta$ (that is, for every $t\in\{0,\ldots,N\}$, $f(t)=\beta$), then we use the shorthand $c=[x_i \si \beta]$. Moreover, we suppose that every constraint represented by a quadruple is associated with two distinct variables.

\medskip
\noindent{\bf Karp Reduction.} 
Given an instance $I=(P,k,\ig,\og,\{\how_x\}|_{x\in {\cal C}(P)\cup\reflex(P)})$ of {\sc Structured Art Gallery}, define an instance $\red(I)=(X,C,N)$ of {\sc Monotone 2-CSP} as follows. 
Let $k^\star=\sum_{e\in {\cal C}(P)\cup\reflex(P)}\ig(e)$, $X=\{x_1,x_2,\ldots,x_{k^\star}\}$ and $N=n+1$. (Here, $n=|V|$.) Additionally, let $\bij$ be an arbitrary bijective function from $X$ to $\{(e,i): e\in {\cal C}(P)\cup\reflex(P),i\in\{1,\ldots,\ig(e)\}\}$. Intuitively, for any variable $x\in X$ with $\bij(x)=(e,i)$, we think of $x$ as the $i^{\mathrm{th}}$ guard to be placed in region $e$. In particular, the value to be assigned to $x$ is the identity of this guard. The values $0$ and $n+1$ are not identities of vertices in $V$, and we will ensure that no solution assignment assigns them; we note that these two values are useful because they will allow us to exclude assignments that should not be solutions.
Next, we define our constraints and show that their functions are monotone.

\medskip
\noindent{\bf Association.} For each $x\in X$ with $\bij(x)=(e,i)$, we need to ensure that the vertex assigned to $x$ is within the region $e$. To this end, we introduce the following constraints.
\begin{itemize}
\item If $e\in\reflex(P)$, then insert the constraint $[x = e]$. (That is, insert $[x\leq e]$ and $[x \geq e]$.)
\item Else, $\bij(x)=(e,j)$ for $e\in{\cal C}(P)$. Let $\ell$ and $h$ be the smallest and largest vertices in $e$, respectively, and insert the constraints $[x \geq \ell]$ and $[x \leq h]$.
\end{itemize}
Let $A$ denote this set of constraints.

\medskip
\noindent{\bf Order in a convex region.} For all $x,x'\in X$ where $\bij(x)=(C,i)$ and $\bij(x')=(C,j)$ for the same convex region $C\in{\cal C}(P)$ and $i<j$, we need to ensure that the vertex assigned to $x'$ is larger than the one assigned to $x$. To this end, we introduce the constraint $[x' \geq f(x)]$ where $f$ is defined as follows. For all $q\in\{0,\ldots,N-1\}$, $f(q)=q+1$, and $f(N)=N$. Let $O$ denote this set of constraints. We note that the constraints in $A\cup O$ together enforce each variable $x\in X$ with $\bij(x)=(C,i)$ for $C\in{\cal C}(P)$ to be assigned the $i^{th}$ guard placed in $C$.

\medskip
\noindent{\bf Guarding reflex vertices.} For every reflex vertex $y\in\reflex(P)$ with $\how_{y}(1)=(e,i)$, we need to ensure that the vertex assigned to $x=\bij^{-1}(e,i)$ sees $y$. To this end, consider two cases. First, suppose that $e\in\reflex(P)$. Then, {\em (i)} if $e$ does not see $y$, output \No, and {\em (ii)} else, no constraint is introduced. Second, suppose that $e\in {\cal C}(P)$. Denote $\ell=\first(y,e)$ and $h=\last(y,e)$. Then, {\em (i)} if $\ell$ (and thus also $h$) is $\nil$, then output \No, and {\em (ii)} else, introduce the constraints $c^1_{y}=[x \geq \ell]$ and $c^2_{y}=[x \leq h]$.

\medskip
\noindent{\bf Guarding first vertices in convex regions.} For every convex region $C=[q,q']\in{\cal C}(P)$ with $\how_C(1)=(e,i)$, we need to ensure that the vertex assigned to $x=\bij^{-1}(e,i)$ sees $q$, the first vertex of $C$. To this end, consider two cases. First, suppose that $e\in\reflex(P)$. Then, {\em (i)} if $e$ does not see $q$, output \No, and {\em (ii)} else, no constraint is introduced. Second, suppose that $e\in {\cal C}(P)$. Denote $\ell=\first(q,e)$ and $h=\last(q,e)$. Then, {\em (i)} if $\ell$ is $\nil$, then output \No, and {\em (ii)} else, insert the constraints $c^1_{(C,1)}=[x \geq \ell]$ and $c^2_{(C,1)}=[x \leq h]$.

\medskip
\noindent{\bf Guarding last vertices in convex regions.} For every convex region $C=[q,q']\in{\cal C}(P)$ with $\how_C(\og(C))=(e,i)$, we need to ensure that the vertex assigned to $x=\bij^{-1}(e,i)$ sees $q'$, the last vertex of $C$. To this end, consider two cases. First, suppose that $e\in\reflex(P)$. Then, {\em (i)} if $e$ does not see $q'$, output \No, and {\em (ii)} else, no constraint is introduced. Second, suppose that $e\in {\cal C}(P)$. Denote $\ell=\first(q',e)$ and $h=\last(q',e)$. Then, {\em (i)} if $\ell$ is $\nil$, then output \No, and {\em (ii)} else, insert the constraints $c^1_{(C,\og(C))}=[x \geq \ell]$ and $c^2_{(C,\og(C))}=[x \leq h]$.

\medskip
\noindent{\bf Guarding middle vertices in convex regions.} For every convex region $C\in{\cal C}(P)$ and $t\in \{2,\ldots,\og(C)\}$, we introduce four constraints based on the following notation.
\begin{itemize}
\item $(e,\gamma)=\how_C(t)$ and $x=\bij^{-1}(e,\gamma)$. Intuitively, the $t^{\mathrm{th}}$ vertex to guard $C$ should be the $\gamma^{\mathrm{th}}$ guard to be placed in $e$, and its precise identity should be assigned to $x$.
If no vertex in $e$ sees at least one vertex in $C$, then return \No.\footnote{In case $e\in\reflex(P)$, we mean that $e$ itself does not see any vertex in $C$.} Let ${\ell}$ and ${h}$ be the smallest and largest vertices in $e$ that see at least one vertex in $C$, respectively.
\item $(e',\gamma')=\how_C(t-1)$ and $x'=\bij^{-1}(e',\gamma')$. Intuitively, the $(t-1)^{\mathrm{th}}$ vertex to guard $C$ should be the ${\gamma'}^{\mathrm{th}}$ guard to be placed in $e'$, and its precise identity should be assigned to $x'$.
If no vertex in $e'$ sees at least one vertex in $C$, then return \No. Let ${\ell'}$ and ${h'}$ be the smallest and largest vertices in $e'$ that see at least one vertex in $C$, respectively.
\end{itemize}

Now, insert the constraints $\widetilde{c}^1_{(C,t)}=[x \geq \ell]$ and $\widetilde{c}^2_{(C,t)}=[x \leq h]$. Intuitively, these two constraints {\em help} to ensure that $x$ will be assigned a vertex that sees at least one vertex in $C$. However, these constraints alone are insufficient for this task---ensuring that we pick a guard between two vertices that see vertices in $C$ does not ensure that this guard sees vertices in $C$.\footnote{For example, in Fig.~\ref{fig:view}, neither $\first(4,[8,19])$ nor $\first(6,[8,19])$ is $\nil$, but $\first(5,[8,19])=\nil$.} Nevertheless, combined with our final constraints, this task is achieved.

Lastly, we consider two sets of four cases. The first set introduces a constraint to ensure that $x$, which stands for the $t^{\mathrm{th}}$ vertex to guard $C$, should satisfy that the first vertex in $C$ seen by $x$ is smaller or equal than the vertex larger by $1$ than the last vertex in $C$ seen by $x'$, which stands for the $(t-1)^{\mathrm{th}}$ vertex to guard $C$. On the other hand, the second set introduces a constraint to ensure that the last vertex in $C$ seen by $x$ is larger than the last vertex in $C$ seen by $x'$. Together, because views have no ``gaps'', this would imply that $x$ sees the vertex in $C$ that is larger by $1$ than the last vertex in $C$ seen by $x'$.
To unify notation, if $e$ (or $e'$) is a reflex vertex, we say that the way $e$ (or $e'$) views $C$ is non-decreasing with respect to both $\first$ and $\last$. Here, Lemma \ref{lem:mainStructF} ensures that at least one case in each set is satisfied.
We start with the first set of four~cases.
\begin{enumerate}
\item The way $e'$ views $C$ is non-decreasing with respect to $\last$, and the way $e$ views $C$ is non-decreasing with respect to $\first$. We insert a constraint $[x \leq f(x')]$, where $f:\{0,\ldots,N\}\rightarrow\{0,\ldots,N\}$ is defined as follows.
	\begin{itemize}
	\item For all $i<\ell'$: $f(i)=0$.
	\item For $i=\ell',\ell'+1,\ldots,h'$: Denote $a=\last({i},C)$.  We have two subcases.
		\begin{itemize}
		\item	If {\bf(i)} $a=\nil$, {\bf(ii)} ${a+1}\notin C$, or {\bf(iii)} $\first(j,C)\leq {a+1}$ for no $j\in e$, then $f(i)=f(i-1)$.  
		\item Otherwise, let $j$ be the largest vertex in $e$ such that $\first(j,C)\leq {a+1}$, and define $f(i)=j$.
		\end{itemize}
	\item For all $i>h'$: $f(i)=N$.
	\end{itemize}
	{\bf Monotonicity.} We claim that $f$ is monotonically non-decreasing. To show this, we choose some $i\in\{1,\ldots,N\}$. If $i\leq\ell'$ or $i>h'$, then it is clear that $f(i)\geq f(i-1)$. Now, suppose that $\ell'<i\leq h'$. If $a=\nil$, ${a+1}\notin C$ or $\first(j,C)\leq {a+1}$ for no $j\in e$, then it is clear that $f(i)\geq f(i-1)$. Hence, we next suppose that this is not the case. Then, $j$ is well-defined. To prove that $f(i)\geq f(i-1)$, we need to show that $f(i-1)\leq j$. Let $\widehat{i}$ be the largest vertex in $\{\ell',\ldots,i-1\}$ such that ${\widehat{a}}=\last({\widehat{i}},C)\neq\nil$, ${\widehat{a}+1}\in C$, and there is a vertex ${\widehat{j}}\in e$ such that $\first({\widehat{j}},C)\leq {\widehat{a}+1}$. (If such a vertex does not exist, then $f(i-1)=0$, and we are done.) Denote $\widehat{j}=f(\widehat{i})$. Note that it suffices to show that $j\geq \widehat{j}$. Because the way $e'$ views $C$ is non-decreasing with respect to $\last$, we have that ${\widehat{a}}\leq a$. Then, because the way $e$ views $C$ is non-decreasing with respect to $\first$, we have that $j\geq \widehat{j}$.

\item The way $e'$ views $C$ is non-decreasing with respect to $\last$, and the way $e$ views $C$ is non-increasing with respect to $\first$. We insert a constraint $[x \geq f(x')]$, where $f:\{0,\ldots,N\}\rightarrow\{0,\ldots,N\}$ is defined as follows.
	\begin{itemize}
	\item For all $i<\ell'$: $f(i)=N$.
	\item For $i=\ell',\ell'+1,\ldots,h'$: Denote $a=\last({i},C)$.  We have two subcases.
		\begin{itemize}
		\item	If {\bf(i)} $a=\nil$, {\bf(ii)} ${a+1}\notin C$, or {\bf(iii)} $\first(j,C)\leq {a+1}$ for no $j\in e$, then $f(i)=f(i-1)$. 
		\item Otherwise, let $j$ be the smallest vertex in $e$ such that $\first(j,C)\leq {a+1}$, and define $f(i)=j$.
		\end{itemize}
	\item For all $i>h'$: $f(i)=0$.
	\end{itemize}
	{\bf Monotonicity.} We claim that $f$ is monotonically non-increasing. To show this, we choose some $i\in\{1,\ldots,N\}$. If $i\leq\ell'$ or $i>h'$, then it is clear that $f(i)\leq f(i-1)$. Now, suppose that $\ell'<i\leq h'$. If $a=\nil$, ${a+1}\notin C$ or $\first(j,C)\leq {a+1}$ for no $j\in e$, then it is clear that $f(i)\leq f(i-1)$. Hence ,we next suppose that this is not the case. Then, $j$ is well-defined. To prove that $f(i)\leq f(i-1)$, we need to show that $f(i-1)\geq j$. Let $\widehat{i}$ be the largest vertex in $\{\ell',\ldots,i-1\}$ such that ${\widehat{a}}=\last({\widehat{i}},C)\neq\nil$, ${\widehat{a}+1}\in C$, and there is a vertex ${\widehat{j}}\in e$ such that $\first({\widehat{j}},C)\leq {\widehat{a}+1}$. (If such a vertex does not exist, then $f(i-1)=N$, and we are done.) Denote $\widehat{j}=f(\widehat{i})$. Note that it suffices to show that $j\leq \widehat{j}$. Because the way $e'$ views $C$ is non-decreasing with respect to $\last$, we have that ${\widehat{a}}\leq a$. Then, because the way $e$ views $C$ is non-increasing with respect to $\first$, we have that $j\leq \widehat{j}$. 

\item The way $e'$ views $C$ is non-increasing with respect to $\last$, and the way $e$ views $C$ is non-decreasing with respect to $\first$. We insert a constraint $[x \leq f(x')]$, where $f:\{0,\ldots,N\}\rightarrow\{0,\ldots,N\}$ is defined as follows.
	\begin{itemize}
	\item For all $i>h'$: $f(i)=0$.\footnote{In the third and fourth cases, unlike the first and second cases, we first define $f$ for integers $i>h'$ rather than for integers $i<\ell'$. The correctness of the reduction relies on this choice of design (we further elaborate on this in footnote~\ref{foot} in the proof).}
	\item For $i=h',h'-1,\ldots,\ell'$: Denote $a=\last({i},C)$.  We have two subcases.
		\begin{itemize}
		\item	If {\bf(i)} $a=\nil$, {\bf(ii)} ${a+1}\notin C$, or {\bf(iii)} $\first(j,C)\leq {a+1}$ for no $j\in e$, then $f(i)=f(i+1)$. 
		\item Otherwise, let $j$ be the largest vertex in $e$ such that $\first(j,C)\leq {a+1}$, and define $f(i)=j$.
		\end{itemize}
	\item For all $i<\ell'$: $f(i)=N$.
	\end{itemize}
	{\bf Monotonicity.} We claim that $f$ is monotonically non-increasing. To show this, we choose some $i\in\{0,\ldots,N-1\}$. If $i<\ell'$ or $i\geq h'$, then it is clear that $f(i)\geq f(i+1)$. Now, suppose that $\ell'\leq i< h'$. If $a=\nil$, ${a+1}\notin C$ or $\first(j,C)\leq {a+1}$ for no $j\in e$, then it is clear that $f(i)\geq f(i+1)$. Hence, we next suppose that this is not the case. Then, $j$ is well-defined. To prove that $f(i)\geq f(i+1)$, we need to show that $j\geq f(i+1)$. Let $\widehat{i}$ be the smallest vertex in $\{i+1,\ldots,h'\}$ such that ${\widehat{a}}=\last({\widehat{i}},C)\neq\nil$, ${\widehat{a}+1}\in C$, and there is a vertex ${\widehat{j}}\in e$ such that $\first({\widehat{j}},C)\leq {\widehat{a}+1}$. (If such a vertex does not exist, then $f(i+1)=0$, and we are done.) Denote $\widehat{j}=f(\widehat{i})$. Note that it suffices to show that $\widehat{j}\leq j$. Because the way $e'$ views $C$ is non-increasing with respect to $\last$, we have that $a\geq {\widehat{a}}$. Then, because the way $e$ views $C$ is non-decreasing with respect to $\first$, we have that $\widehat{j}\leq j$.

\item The way $e'$ views $C$ is non-increasing with respect to $\last$, and the way $e$ views $C$ is non-increasing with respect to $\first$. We insert a constraint $[x \geq f(x')]$, where $f:\{0,\ldots,N\}\rightarrow\{0,\ldots,N\}$ is defined as follows.
	\begin{itemize}
	\item For all $i>h'$: $f(i)=N$.	
	\item For $i=h',h'-1,\ldots,\ell'$: Denote $a=\last({i},C)$.  We have two subcases.
		\begin{itemize}
		\item	If {\bf(i)} $a=\nil$, {\bf(ii)} ${a+1}\notin C$, or {\bf(iii)} $\first(j,C)\leq {a+1}$ for no $j\in e$, then $f(i)=f(i+1)$. 
		\item Otherwise, let $j$ be the smallest vertex in $e$ such that $\first(j,C)\leq {a+1}$, and define $f(i)=j$.
		\end{itemize}
	\item For all $i<\ell'$: $f(i)=0$.
	\end{itemize}
	{\bf Monotonicity.} We claim that $f$ is monotonically non-decreasing. To show this, we choose some $i\in\{0,\ldots,N-1\}$. If $i\geq h'$ or $i<\ell'$, then it is clear that $f(i+1)\geq f(i)$. Now, suppose that $\ell'\leq i<h'$. If $a=\nil$, ${a+1}\notin C$ or $\first(j,C)\leq {a+1}$ for no $j\in e$, then it is clear that $f(i+1)\geq f(i)$. Hence, we next suppose that this is not the case. Then, $j$ is well-defined. To prove that $f(i+1)\geq f(i)$, we need to show that $j\leq f(i-1)$. Let $\widehat{i}$ be the smallest vertex in $\{i+1,\ldots,h'\}$ such that ${\widehat{a}}=\last({\widehat{i}},C)\neq\nil$, ${\widehat{a}+1}\in C$, and there is a vertex ${\widehat{j}}\in e$ such that $\first({\widehat{j}},C)\leq {\widehat{a}+1}$. (If such a vertex does not exist, then $f(i+1)=N$, and we are done.) Denote $\widehat{j}=f(\widehat{i})$. Note that it suffices to show that $\widehat{j}\geq j$. Because the way $e'$ views $C$ is non-increasing with respect to $\last$, we have that $a\geq {\widehat{a}}$. Then, because the way $e$ views $C$ is non-increasing with respect to $\first$, we have that $\widehat{j}\geq j$. 
\end{enumerate}

Let us now give the second set of four cases. Here, each proof of monotonicity follows from arguments similar to those given for the first set, and therefore it is omitted.
\begin{enumerate}
\item The ways $e'$ and $e$ view $C$ are both non-decreasing with respect to $\last$. We insert a constraint $[x \geq f(x')]$, where $f:\{0,\ldots,N\}\rightarrow\{0,\ldots,N\}$ is defined as follows.
	\begin{itemize}
	\item For all $i>h'$: $f(i)=N$.
	\item For $i=h',h'-1,\ldots,\ell'$: Denote $a=\last({i},C)$.  We have two subcases.
		\begin{itemize}
		\item If {\bf (i)} $a=\nil$, {\bf (ii)} ${a+1}\notin C$, or {\bf (iii)} $\last(j,C)\geq {a+1}$ for no $j\in e$, then $f(i)=f(i+1)$.
		\item Otherwise, let $j$ be the smallest vertex in $e$ such that $\last(j,C)\geq {a+1}$, and define $f(i)=j$.
		\end{itemize}
	\item For all $i<\ell'$: $f(i)=0$.
	\end{itemize}

\item The way $e'$ views $C$ is non-decreasing with respect to $\last$, and the way $e$ views $C$ is non-increasing with respect to $\last$. We insert a constraint $[x \leq f(x')]$, where $f:\{0,\ldots,N\}\rightarrow\{0,\ldots,N\}$ is defined as follows.
	\begin{itemize}
	\item For all $i>h'$: $f(i)=0$.	
	\item For $i=h',h'-1,\ldots,\ell'$: Denote $a=\last({i},C)$.  We have two subcases.
		\begin{itemize}
		\item If {\bf (i)} $a=\nil$, {\bf (ii)} ${a+1}\notin C$, or {\bf (iii)} $\last(j,C)\geq {a+1}$ for no $j\in e$, then $f(i)=f(i+1)$. 
		\item Otherwise, let $j$ be the largest vertex in $e$ such that $\last(j,C)\geq {a+1}$, and define $f(i)=j$.
		\end{itemize}
	\item For all $i<\ell'$: $f(i)=N$.
	\end{itemize}

\item The way $e'$ views $C$ is non-increasing with respect to $\last$, and the way $e$ views $C$ is non-decreasing with respect to $\last$. We insert a constraint $[x \geq f(x')]$, where $f:\{0,\ldots,N\}\rightarrow\{0,\ldots,N\}$ is defined as follows.
	\begin{itemize}
	\item For all $i<\ell'$: $f(i)=N$.
	\item For $i=\ell',\ell'+1,\ldots,h'$: Denote $a=\last({i},C)$.  We have two subcases.
		\begin{itemize}
		\item If {\bf (i)} $a=\nil$, {\bf (ii)} ${a+1}\notin C$, or {\bf (iii)} $\last(j,C)\geq {a+1}$ for no $j\in e$, then $f(i)=f(i-1)$. 
		\item Otherwise, let $j$ be the smallest vertex in $e$ such that $\last(j,C)\geq {a+1}$, and define $f(i)=j$.
		\end{itemize}
	\item For all $i>h'$: $f(i)=0$.
	\end{itemize}

\item The ways $e'$ and $e$ view $C$ are both non-increasing with respect to $\last$. We insert a constraint $[x \leq f(x')]$, where $f:\{0,\ldots,N\}\rightarrow\{0,\ldots,N\}$ is defined as follows.
	\begin{itemize}
	\item For all $i<\ell'$: $f(i)=0$.
	\item For $i=\ell',\ell'+1,\ldots,h'$: Denote $a=\last({i},C)$.  We have two subcases.
		\begin{itemize}
		\item If {\bf (i)} $a=\nil$, {\bf (ii)} ${a+1}\notin C$, or {\bf (iii)} $\last(j,C)\geq {a+1}$ for no $j\in e$, then $f(i)=f(i-1)$.
		\item Otherwise, let $j$ be the largest vertex in $e$ such that $\last(j,C)\geq {a+1}$, and define $f(i)=j$.
		\end{itemize}
	\item For all $i>h'$: $f(i)=N$.
	\end{itemize}
\end{enumerate}

The following observation directly follows from the definition of our reduction.

\begin{obs}\label{obs:reduction}
For an instance $I=(P,k,\ig,\og,\{\how_x\}|_{x\in {\cal C}(P)\cup\reflex(P)})$ of {\sc Structured Art Gallery}, $|X|=\OO(r)$ where $\red(I)=(X,C,N)$. Moreover, $\red$ is computable in polynomial time.
\end{obs}

To establish the correctness of our reduction, we start with the reverse direction.

\begin{lem}\label{lem:reverse}
Let $I=(P,k,\ig,\og,\{\how_x\}|_{x\in {\cal C}(P)\cup\reflex(P)})$ be an instance of {\sc Structured Art Gallery}, and denote $\red(I)=(X,C,N)$. If $(X,C,N)$ is a \Yes-instance of {\sc Monotone 2-CSP}, then $I$ is a \Yes-instance of {\sc Structured Art Gallery}.
\end{lem}

\begin{proof} Suppose that $(X,C,N)$ is a \Yes-instance of {\sc Monotone 2-CSP}. Accordingly, let $\alpha: X\rightarrow\{0,\ldots,N\}$ be a solution to $(X,C,N)$. By the constraints in $A$, we have that for all $x\in X$, for $(e,i)=\bij^{-1}(x)$, it holds that ${\alpha(x)}\in e$.\footnote{If $e\in\reflex(P)$, by ${\alpha(x)}\in e$ we mean ${\alpha(x)}=e$.} In particular, for $S=\{{\alpha(x)}: x\in X\}$, we have that $S\subseteq V$. In what follows, we show that $S$ is a solution to $I$, which would conclude the proof. Because $|X|\leq k$, we immediately have that $|S|\leq k$. Thus, it remains to show that Conditions \ref{condition:structured1}, \ref{condition:structured2} and \ref{condition:structured3} in the definition of the objective of {\sc Structured Art Gallery} are satisfied.

\medskip
{\noindent\bf Condition \ref{condition:structured1}.}
First, note that for each convex region or reflex vertex $y\in {\cal C}(P)\cup\reflex(P)$, $|S\cap y|=|\{x\in X: (y,i)=\bij^{-1}(x)$ for some $i\in\{1,\ldots,\ig(y)\}\}|=\ig(y)$. Here, the first equality followed from the definition of $S$, and the last equality followed from the the fact that $\bij$ is bijective.
Accordingly, for each $y\in {\cal C}(P)\cup\reflex(P)$ and $i\in\{1,\ldots,\ig(y)\}$, let $s_{(y,i)}$ denote the $i^{th}$ largest vertex in $S\cap y$; by the constraints in $A\cup O$, we have that $s_{(y,i)}={\alpha(x)}$ for $x=\bij^{-1}(y,i)$.

\medskip
{\noindent\bf Condition \ref{condition:structured2}.}
Consider some reflex vertex $y\in\reflex(P)$, and denote $(e,i)=\how_{(y,1)}$. 
First, suppose that $e\in\reflex(P)$. Then, $e$ sees $y$, else we would have outputted \No. By the constraints in $A$, we have that $e=s_{\how_y(1)}\in S$, and hence $s_{\how_y(1)}\in S$ sees $y$. Second, suppose that $e\in {\cal C}(P)$. Then, since $\alpha$ satisfies the constraints $c^1_{y}$ and $c^2_{y}$, for the variable $x\in X$ that satisfies $\bij(x)=(e,i)$, we have that $\first(y,e)\leq {\alpha(x)}\leq \last(y,e)$. By Lemma \ref{lem:seeAllBetFin}, this means that ${\alpha(x)}$ sees $y$. Thus, because $s_{\how_y(1)}=s_{(e,i)}={\alpha(x)}$, we have that $s_{\how_y(1)}$ sees $y$.

\medskip
{\noindent\bf Condition \ref{condition:structured3a}.}
In what follows, consider some convex region $C\in {\cal C}(P)$. Here, we need to show that $\first(s_{\how_C(1)},C)$ is the smallest vertex in $C$. Denote $(e,i)=\how_C(1)$ and $x=\bij^{-1}(e,i)$. Additionally, denote the first vertex in $C$ by $q$. First, suppose that $e\in\reflex(P)$. Then, $e$ sees $q$, else we would have outputted \No. By the constraints in $A$, we have that $e=s_{\how_C(1)}\in S$. Thus, $s_{\how_C(1)}$ sees $q$ (which means that $\first(s_{\how_C(1)},C)$ is the smallest vertex in $C$).
Second, suppose that $e\in{\cal C}$. Let $\ell=\first(q,e)$ and $h=\last(q,e)$. If $\ell$ (and $h$) is $\nil$, then we would have outputted \No. Thus, by the constraints $c^1_{(C,1)}$ and $c^2_{(C,1)}$, we have that $\ell\leq {\alpha(x)}\leq h$. By Lemma \ref{lem:seeAllBetFin}, this means that ${\alpha(x)}$ sees $q$. Thus, because $s_{\how_C(1)}=s_{(e,i)}={\alpha(x)}$, we have that $s_{\how_C(1)}$ sees $q$.

\medskip
{\noindent\bf Condition \ref{condition:structured3c}.}
Here, we need to show that $\last(s_{\how_C(\og(C))},C)$ is the largest vertex in $C$.  Denote $(e,i)=\how_C(\og(C))$ and $x=\bij^{-1}(e,i)$. Additionally, denote the last vertex in $C$ by $q$. First, suppose that $e\in\reflex(P)$. Then, $e$ sees $q$, else we would have outputted \No. By the constraints in $A$, we have that $e=s_{\how_C(\og(C)}\in S$. Thus, $s_{\how_C(\og(C)}$ sees $q$ (which means that $\last(s_{\how_C(\og(C))},C)$ is the largest vertex in $C$). Second, suppose that $e\in{\cal C}$. Let $\ell=\first(q,e)$ and $h=\last(q,e)$. If $\ell$ (and $h$) is $\nil$, then we would have outputted \No. Thus, by the constraints $c^1_{(C,\og(C))}$ and $c^2_{(C,\og(C))}$, we have that $\ell\leq {\alpha(x)}\leq h$. By Lemma \ref{lem:seeAllBetFin}, this means that ${\alpha(x)}$ sees $q$. Thus, because $s_{\how_C(\og(C))}=s_{(e,i)}={\alpha(x)}$, we have that $s_{\how_C(\og(C))}$ sees $q$.

\medskip
{\noindent\bf Condition \ref{condition:mid}.} 
Lastly, we need to show that for every $t\in \{1,\ldots,\og(C)-1\}$, it holds that
\[\first(s_{\how_C(t+1)},C)-1\leq\last(s_{\how_C(t)},C)\leq\last(s_{\how_C(t+1)},C)-1.\]
Rephrased differently, we need to show that for every $t\in \{2,\ldots,\og(C)\}$, it holds that
\[\first(s_{\how_C(t)},C)-1\leq\last(s_{\how_C(t-1)},C)\leq\last(s_{\how_C(t)},C)-1.\]
Observe that these inequalities encompass the requirement that $s_{\how_C(t)}$ sees at least one vertex in $C$. (Indeed, $1$ cannot be subtracted from $\nil$, and $\nil$ cannot be smaller or larger than an integer.) For $t=1$, we only claim that $s_{\how_C(1)}$ sees at least one vertex in $C$. Now, the proof is by induction on $t$.\footnote{Here, induction is not mandatory. Instead, we can rely on the constraints marked with a tilde. However, these constraints are required for a different purpose (rather than only to encompass the inductive hypothesis). To highlight this, we prefer to use induction.} In the basis, where $t=1$, the claim holds since we have already proved that Condition \ref{condition:structured3a} is satisfied. Next, we suppose that the claim is true for all $t'\in\{1,\ldots,t-1\}$, and prove it for $t\in \{2,\ldots,\og(C)\}$

Denote $(e,\gamma)=\how_C(\og(t))$ and $x=\bij^{-1}(e,\gamma)$. In addition, denote $(e',\gamma')=\how_C(\og(t-1))$ and $x'=\bij^{-1}(e',\gamma')$. By the constraints in $A\cup O$, we have that $s_{\how_C(\og(t))}=s_{(e,\gamma)}={\alpha(x)}$ and $s_{\how_C(\og(t-1))}=s_{(e',\gamma')}={\alpha(x')}$. Denote $a=\last({\alpha(x')},C)$, and observe that $a\neq\nil$ by the inductive hypothesis. With this notation, our task is to show that {\em (i)} $a\geq b-1$ for $b=\first({\alpha(x)},C)$, and {\em (ii)} $a\leq q-1$ for $q=\last({\alpha(x)},C)$. If $a+1\notin C$, then the second condition cannot be satisfied. Therefore, it suffices to show that
\begin{enumerate}
\item\label{condition:ab} either $a+1\in C$ or $a\geq b-1$ for $b=\first({\alpha(x)},C)$, and 
\item\label{condition:aq}  $a\leq q-1$ for $q=\last({\alpha(x)},C)$.
\end{enumerate}

The first set of four cases\footnote{See ``guarding the middle vertices in a convex region'' in Section \ref{sec:reduction}.} is necessary mainly to prove the first condition above, and the second set of four cases is necessary mainly to prove the second condition above. However, to rule out the possibility that $b=q=\nil$, the first set of four cases is also required to prove the second condition, and the second set of four cases is also required to prove the first one. Thus, both conditions are proved simultaneously. In this context, let $c=[x \si f(x')]$ be the constraint that was introduced due to appropriate case from the first set of four cases, and let $\widehat{c}=[x \widehat{\si} \widehat{f}(x')]$ be the constraint that was introduced due to the appropriate case from the second set of four cases.
We consider eight cases, depending on the way $e'$ views $C$ with respect to $\last$, and the way $e$ views $C$ with respect to both $\first$ and $\last$.

\medskip
\noindent{\bf Case 1 of First Set.} In this case, we suppose that the way $e'$ views $C$ is non-decreasing with respect to $\last$, and the way $e$ views $C$ is non-decreasing with respect to $\first$. Then, $\si$ is equal to $\leq$. Moreover, in this case, $f(\alpha(x'))$ is defined as follows. (Here, recall that the possibility that $a=\nil$ has already been ruled out.)
If ${a+1}\notin C$ or $\first(j,C)\leq {a+1}$ for no $j\in e$, then $f(\alpha(x'))=f(\alpha(x')-1)$. Otherwise, $f(\alpha(x'))$ is the largest vertex $j\in e$ such that $\first(j,C)\leq {a+1}$. In what follows, we suppose that $a+1\in C$ for the sake of the proof of Condition \ref{condition:ab}, else the proof of this condition is complete.

Since $\alpha$ is a solution to $(X,C,N)$, we have that $\alpha(x)\leq f(\alpha(x'))$. In particular, since ${\alpha(x)}\notin\{0,N\}$ (because ${\alpha(x)}\in S$ and $S\subseteq V$), we have that $f(\alpha(x'))\neq 0$. To proceed our analysis, we define $\delta$ and $a^\star$ as follows. Let $\delta$ be the largest vertex, not larger than $\alpha(x')$, such that $f(\delta)=f(\alpha(x'))$ and the following conditions hold for ${a}^\star=\last({\delta},C)$:
	\begin{enumerate}
	\item ${a}^\star\neq\nil$ and ${a}^\star+1\in C$;
	\item ${f(\alpha(x'))}$ is the largest vertex $v\in e$ such that $\first(v,C)\leq {a^\star+1}$.
	\end{enumerate}
The existence of such $\delta$ follows from the definition of $f$ and because $f(\alpha(x'))\neq 0$. Since $\delta\leq\alpha(x')$ and the way $e'$ views $C$ is non-decreasing with respect to $\last$, we have that ${a}^\star\leq a$. Thus, $\first({f(\alpha(x'))},C)\leq {a}^\star+1\leq  {a+1}$. By the definition of $f(\alpha(x'))$, this means that $f(\alpha(x'))$ is the largest vertex $j\in e$ such that $\first(j,C)\leq {a+1}$.  Because $\alpha(x)\leq f(\alpha(x'))=j$ and the way $e$ views $C$ is non-decreasing with respect to $\first$, we have that either $\first({\alpha(x)},C)\leq \first(j,C)$ or $\first({\alpha(x)},C)=\nil$. In the first scenario, $b\leq a+1$, hence the proof of Condition \ref{condition:ab} is complete. (The second scenario is addressed ahead.)

\medskip
\noindent{\bf Case 1 of First Set + Case 1 of Second Set.} In this case, we suppose that $e$ views $C$ is non-decreasing with respect to $\last$. Then, $\widehat{\si}$ is equal to $\geq$. Moreover, in this case, $\widehat{f}(\alpha(x'))$ is defined as follows. (Here, recall that the possibility that $a=\nil$ has already been ruled out.)
If ${a+1}\notin C$ or $\last({\widehat{j}},C)\geq {a+1}$ for no $\widehat{j}\in e$, then $\widehat{f}(\alpha(x'))=\widehat{f}(\alpha(x')+1)$. Otherwise, $\widehat{f}(\alpha(x'))$ is the smallest vertex ${\widehat{j}}\in e$ such that $\last({\widehat{j}},C)\geq {a+1}$.

Since $\alpha$ is a solution to $(X,C,N)$, we have that $\alpha(x)\geq \widehat{f}(\alpha(x'))$. In particular, since ${\alpha(x)}\notin\{0,N\}$ (because ${\alpha(x)}\in S$ and $S\subseteq V$), we have that $\widehat{f}(\alpha(x'))\neq N$. To proceed our analysis, we define $\widehat{\delta}$ and $\widehat{a}^\star$ as follows. Let $\widehat{\delta}$ be the smallest vertex, not smaller than $\alpha(x')$, such that $\widehat{f}(\widehat{\delta})=\widehat{f}(\alpha(x'))$ and the following conditions hold for $\widehat{a}^\star=\last({\delta},C)$:
	\begin{enumerate}
	\item $\widehat{a}^\star\neq\nil$ and $\widehat{a}^\star+1\in C$;
	\item ${\widehat{f}(\alpha(x'))}$ is the smallest vertex $\widehat{v}\in e$ such that $\last(\widehat{v},C)\geq {\widehat{a}^\star+1}$.
	\end{enumerate}
The existence of such $\widehat{\delta}$ follows from the definition of $\widehat{f}$ and because $\widehat{f}(\alpha(x'))\neq N$.\footnote{If the function $f$ were defined first for $i<\ell'$ rather than for $i>h'$, then the existence of $\widehat{\delta}$ would not have followed. Specifically, we need the integer that ``propagates'' in the definition of $\widehat{f}$ to be $N$ rather than $0$ because we have the assertion $\alpha(x)\geq \widehat{f}(\alpha(x'))$ rather than $\alpha(x)\leq \widehat{f}(\alpha(x'))$.\label{foot}} Since $\widehat{\delta}\geq\alpha(x')$ and the way $e'$ views $C$ is non-decreasing with respect to $\last$, we have that $\widehat{a}^\star\geq a$. Thus, $\last({\widehat{f}(\alpha(x'))},C)\geq \widehat{a}^\star+1\geq  {a+1}$, and hence $a+1\in C$.
By the definition of $\widehat{f}(\alpha(x'))$, this means that $\widehat{f}(\alpha(x'))$ is the smallest vertex $\widehat{j}\in e$ such that $\last({\widehat{j}},C)\geq {a+1}$. Because $\alpha(x)\geq \widehat{f}(\alpha(x'))=\widehat{j}$ and the way $e$ views $C$ is non-decreasing with respect to $\last$, we have that either $\last({\alpha(x)},C)\geq \last(\widehat{j},C)$ or $\last({\alpha(x)},C)=\nil$. In the first case, $q\geq a+1$, hence the proof of Condition \ref{condition:aq} is complete.

We are left with the scenario where $\first({\alpha(x)},C)=\last({\alpha(x)},C)=\nil$. To handle this scenario, recall that $\widehat{j}\leq\alpha(x)\leq j$, and $\first(j,C)\leq a+1\leq \last(\widehat{j},C)$. Because the way $e$ views $C$ is non-decreasing with respect to both $\first$ and $\last$, the first chain of inequalities implies that $\first(\widehat{j},C)\leq \first(j,C)$ and $\last(\widehat{j},C)\leq \last(j,C)$. Thus,  $\first(j,C)\leq a+1\leq \last(j,C)$ and $\first(\widehat{j},C)\leq a+1\leq \last(\widehat{j},C)$. By Lemma \ref{lem:seeAllBetFin}, we have that both $j$ and $\widehat{j}$ see $a+1$. In turn, by Lemma \ref{lem:seeAllBetFin} and since $\widehat{j}\leq\alpha(x)\leq j$, this means that $\alpha(x)$ sees $a+1$, which is a contradiction to $\first({\alpha(x)},C)=\last({\alpha(x)},C)=\nil$. Thus, this scenario cannot occur.

\medskip
\noindent{\bf Case 1 of First Set + Case 2 of Second Set.} In this case, we suppose that the way $e$ views $C$ is non-increasing with respect to $\last$. Then, $\widehat{\si}$ is equal to $\leq$. Moreover, in this case, $\widehat{f}(\alpha(x'))$ is defined as follows. (Here, recall that the possibility that $a=\nil$ has already been ruled out.)
If ${a+1}\notin C$ or $\last({\widehat{j}},C)\geq {a+1}$ for no $\widehat{j}\in e$, then $\widehat{f}(\alpha(x'))=\widehat{f}(\alpha(x')+1)$. Otherwise, $\widehat{f}(\alpha(x'))$ is the largest vertex ${\widehat{j}}\in e$ such that $\last({\widehat{j}},C)\geq {a+1}$.

Since $\alpha$ is a solution to $(X,C,N)$, we have that $\alpha(x)\leq \widehat{f}(\alpha(x'))$. In particular, since ${\alpha(x)}\notin\{0,N\}$ (because ${\alpha(x)}\in S$ and $S\subseteq V$), we have that $\widehat{f}(\alpha(x'))\neq 0$. To proceed our analysis, we define $\widehat{\delta}$ and $\widehat{a}^\star$ as follows. Let $\widehat{\delta}$ be the smallest vertex, not smaller than $\alpha(x')$, such that $\widehat{f}(\widehat{\delta})=\widehat{f}(\alpha(x'))$ and the following conditions hold for $\widehat{a}^\star=\last({\delta},C)$:
	\begin{enumerate}
	\item $\widehat{a}^\star\neq\nil$ and $\widehat{a}^\star+1\in C$;
	\item ${\widehat{f}(\alpha(x'))}$ is the largest vertex $\widehat{v}\in e$ such that $\last(\widehat{v},C)\geq {\widehat{a}^\star+1}$.
	\end{enumerate}
The existence of such $\widehat{\delta}$ follows from the definition of $\widehat{f}$ and because $\widehat{f}(\alpha(x'))\neq 0$. Since $\widehat{\delta}\geq\alpha(x')$ and the way $e'$ views $C$ is non-decreasing with respect to $\last$, we have that $\widehat{a}^\star\geq a$. Thus, $\last({\widehat{f}(\alpha(x'))},C)\geq \widehat{a}^\star+1\geq  {a+1}$, and hence $a+1\in C$.
By the definition of $\widehat{f}(\alpha(x'))$, this means that $\widehat{f}(\alpha(x'))$ is the largest vertex $\widehat{j}\in e$ such that $\last({\widehat{j}},C)\geq {a+1}$. Because $\alpha(x)\leq \widehat{f}(\alpha(x'))=\widehat{j}$ and the way $e$ views $C$ is non-increasing with respect to $\last$, we have that either $\last({\alpha(x)},C)\geq \last(\widehat{j},C)$ or $\last({\alpha(x)},C)=\nil$. In the first case, $q\geq a+1$, hence the proof of Condition \ref{condition:aq} is complete.

We are left with the scenario where $\first({\alpha(x)},C)=\last({\alpha(x)},C)=\nil$. To handle this scenario, recall that $\alpha(x)\leq \min(\widehat{j},j)$. 
Due to the constraint $\widetilde{c}^1_{(C,t)}=[x \geq \ell]$, we have that $\ell\leq \alpha(x)$, and therefore $\ell\leq \min(\widehat{j},j)$. Moreover, by the definition of $\ell$, it sees at least one vertex in $C$. Thus, since the way $e$ views $C$ is non-decreasing with respect to $\first$ and non-increasing with respect to $\last$, we have that $\first(\ell,C)\leq \first(j,C)\leq \last(j,C)\leq \last(\ell,C)$. By Lemma \ref{lem:seeAllBetFin}, this means that $\ell$ sees $\first(j,C)$. In turn, by Lemma \ref{lem:seeAllBetFin} and since $\ell\leq\alpha(x)\leq j$, this means that $\first(j,C)$ sees $\alpha(x)$, which is a contradiction to $\first({\alpha(x)},C)=\last({\alpha(x)},C)=\nil$. Thus, this scenario cannot occur.

\medskip
{\bf\noindent The proofs of the other three cases follow the same lines as the proof of the first case. For the sake of illustration, we give the details of the second case.}

\medskip
\noindent{\bf Case 2 of First Set.}  In this case, we suppose that the way $e'$ views $C$ is non-decreasing with respect to $\last$, and the way $e$ views $C$ is non-increasing with respect to $\first$. Then, $\si$ is equal to $\geq$. Moreover, in this case, $f(\alpha(x'))$ is defined as follows. (Here, recall that the possibility that $a=\nil$ has already been ruled out.)
If ${a+1}\notin C$ or $\first(j,C)\leq {a+1}$ for no $j\in e$, then $f(\alpha(x'))=f(\alpha(x')-1)$. Otherwise, $f(\alpha(x'))$ is the smallest vertex $j\in e$ such that $\first(j,C)\leq {a+1}$. In what follows, we suppose that $a+1\in C$ for the sake of the proof of Condition \ref{condition:ab}, else the proof of this condition is complete.

Since $\alpha$ is a solution to $(X,C,N)$, we have that $\alpha(x)\geq f(\alpha(x'))$. In particular, since ${\alpha(x)}\notin\{0,N\}$ (because ${\alpha(x)}\in S$ and $S\subseteq V$), we have that $f(\alpha(x'))\neq N$. To proceed our analysis, we define $\delta$ and $a^\star$ as follows. Let $\delta$ be the largest vertex, not larger than $\alpha(x')$, such that $f(\delta)=f(\alpha(x'))$ and the following conditions hold for ${a}^\star=\last({\delta},C)$:
	\begin{enumerate}
	\item ${a}^\star\neq\nil$ and ${a}^\star+1\in C$;
	\item ${f(\alpha(x'))}$ is the smallest vertex $v\in e$ such that $\first(v,C)\leq {a^\star+1}$.
	\end{enumerate}
The existence of such $\delta$ follows from the definition of $f$ and because $f(\alpha(x'))\neq N$. Since $\delta\leq\alpha(x')$ and the way $e'$ views $C$ is non-decreasing with respect to $\last$, we have that ${a}^\star\leq a$. Thus, $\first({f(\alpha(x'))},C)\leq {a}^\star+1\leq  {a+1}$. By the definition of $f(\alpha(x'))$, this means that $f(\alpha(x'))$ is the smallest vertex $j\in e$ such that $\first(j,C)\leq {a+1}$.  Because $\alpha(x)\geq f(\alpha(x'))=j$ and the way $e$ views $C$ is non-increasing with respect to $\first$, we have that either $\first({\alpha(x)},C)\leq \first(j,C)$ or $\first({\alpha(x)},C)=\nil$. In the first scenario, $b\leq a+1$, hence the proof of Condition \ref{condition:ab} is complete.

\medskip
\noindent{\bf Case 2 of First Set + Case 1 of Second Set.} In this case, we suppose that $e$ views $C$ is non-decreasing with respect to $\last$. Then, $\widehat{\si}$ is equal to $\geq$. By repeating the {\em exact} same arguments given in ``Case 1 of First Set + Case 1 of Second Set'', we derive that either $\last({\alpha(x)},C)\geq \last(\widehat{j},C)$ or $\last({\alpha(x)},C)=\nil$. Indeed, all the arguments presented up to that point are oblivious to the way in which $e$ views $C$ with respect to $\first$.
In the first case (where $\last({\alpha(x)},C)\geq \last(\widehat{j},C)$), $q\geq a+1$, hence the proof of Condition \ref{condition:aq} is complete.

We are left with the scenario where $\first({\alpha(x)},C)=\last({\alpha(x)},C)=\nil$. To handle this scenario, recall that $\max(\widehat{j},j)\leq\alpha(x)$. Due to the constraint $\widetilde{c}^2_{(C,t)}=[x \leq h]$, we have that $\alpha(x)\leq h$, and therefore $\max(\widehat{j},j)\leq h$. Moreover, by the definition of $h$, it sees at least one vertex in $C$. Thus, since the way $e$ views $C$ is non-increasing with respect to $\first$ and non-decreasing with respect to $\last$, we have that $\first(h,C)\leq \first(j,C)\leq \last(j,C)\leq \last(h,C)$. By Lemma \ref{lem:seeAllBetFin}, this means that $h$ sees $\first(j,C)$. In turn, by Lemma \ref{lem:seeAllBetFin} and since $j\leq\alpha(x)\leq h$, this means that $\first(j,C)$ sees $\alpha(x)$, which is a contradiction to $\first({\alpha(x)},C)=\last({\alpha(x)},C)=\nil$. Thus, this scenario cannot occur.

\medskip
\noindent{\bf Case 2 of First Set + Case 2 of Second Set.} In this case, we suppose that $e$ views $C$ is non-increasing with respect to $\last$. By repeating the {\em exact} same arguments given in ``Case 1 of First Set + Case 2 of Second Set'', we derive that either $\last({\alpha(x)},C)\geq \last(\widehat{j},C)$ or $\last({\alpha(x)},C)=\nil$. Indeed, all the arguments presented up to that point are oblivious to the way in which $e$ views $C$ with respect to $\first$.
In the first case (where $\last({\alpha(x)},C)\geq \last(\widehat{j},C)$), $q\geq a+1$, hence the proof of Condition \ref{condition:aq} is complete.

We are left with the scenario where $\first({\alpha(x)},C)=\last({\alpha(x)},C)=\nil$. To handle this scenario, recall that $j\leq\alpha(x)\leq \widehat{j}$, and $\first(j,C)\leq a+1\leq \last(\widehat{j},C)$. Because the way $e$ views $C$ is non-increasing with respect to both $\first$ and $\last$, the first chain of inequalities implies that $\first(\widehat{j},C)\leq \first(j,C)$ and $\last(\widehat{j},C)\leq \last(j,C)$. Thus, $\first(j,C)\leq a+1\leq \last(j,C)$ and $\first(\widehat{j},C)\leq a+1\leq \last(\widehat{j},C)$. By Lemma \ref{lem:seeAllBetFin}, we have that both $j$ and $\widehat{j}$ see $a+1$. In turn, by Lemma \ref{lem:seeAllBetFin} and since $j\leq\alpha(x)\leq \widehat{j}$, this means that $\alpha(x)$ sees $a+1$, which is a contradiction to $\first({\alpha(x)},C)=\last({\alpha(x)},C)=\nil$. Thus, this scenario cannot occur.
%
%
%
%
%
\end{proof}

Now, we prove the correctness of the forward direction.

\begin{lem}\label{lem:forward}
Let $I=(P,k,\ig,\og,\{\how_x\}|_{x\in {\cal C}(P)\cup\reflex(P)})$ be an instance of {\sc Structured Art Gallery}, and denote $\red(I)=(X,C,N)$. If $I$ is a \Yes-instance of {\sc Structured Art Gallery}, then $(X,C,N)$ is a \Yes-instance of {\sc Monotone 2-CSP}.
\end{lem}

\begin{proof}
Suppose that $I$ is a \Yes-instance of {\sc Structured Art Gallery}. Accordingly, let $S\subseteq V$ be a solution to $I$. Then, $|S|\leq k$, and the following conditions hold:
\begin{enumerate}
\item For each $y\in {\cal C}(P)\cup\reflex(P)$, $|S\cap y|=\ig(y)$. Accordingly, for each $y\in {\cal C}(P)\cup\reflex(P)$ and $i\in\{1,\ldots,\ig(y)\}$, let $s_{(y,i)}$ denote the $i^{th}$ largest vertex in $S\cap y$.
\item For each $y\in\reflex(P)$, $s_{\how_y(1)}$ sees $y$.
\item For each $C\in {\cal C}(P)$, the following conditions hold:
	\begin{enumerate}
	\item $\first(s_{\how_C(1)},C)$ is the smallest vertex in $C$.
	\item For every $t\in \{1,\ldots,\og(C)-1\}$, denote $a=\last(s_{\how_C(t)},C)$, $j=\first(s_{\how_C(t+1)},C)$ and $q=\last(s_{\how_C(t+1)},C)$. Then, {\em (i)} $a\geq j-1$, and {\em (ii)} $a\leq q-1$.
	\item $\last(s_{\how_C(\og(C))},C)$ is the largest vertex in $C$.
	\end{enumerate}
\end{enumerate}

In order to define an assignment $\alpha: X\rightarrow\{0,\ldots,N\}$, let $x\in X$. Denote $\bij(x)=(e,i)$. Accordingly, let $t$ denote the $i^{th}$ largest vertex $t$ in $S\cap e$, namely, $s_{(e,i)}$. Then, define $\alpha(x)=t$. Since for $e\in {\cal C}(P)\cup\reflex(P)$, $|S\cap e|=\ig(e)$, and by the definition of the bijection $\bij$, we have that $t$ is well-defined. In what follows, we argue that $\alpha$ is a solution to $(X,C,N)$. First, by the definition of $\alpha$, it is clear that all of the constraints in $A\cup O$ are satisfied.

\medskip
{\noindent\bf Guarding reflex vertices.} Consider some $y\in\reflex(P)$. Note that $s_{\how_y(1)}$ sees $y$. Denote $(e,i)=\how_y{(1)}$. If $e\in\reflex(P)$, then $e$ sees $y$ and no constraint is introduced. Next, suppose that $e\in{\cal C}(P)$.  Let $x\in X$ be the variable that satisfies $\bij(x)=\how_{y}(1)$. Denote $\ell=\first(y,e)$ and $h=\last(y,e)$. Since $s_{\how_y(1)}$ sees $y$, neither $\ell$ nor $h$ is $\nil$. We thus have the constraints $c^1_{y}=[x \geq \ell]$ and $c^2_{y}=[x \leq h]$. To prove that $\alpha$ satisfies them, we need to show that $\ell\leq \alpha(x)\leq h$. However, this directly follows from the fact that $\alpha(x)=s_{\how_y(1)}$ sees $y$.

In what follows, we consider some $C\in{\cal C}(P)$, and show that $\alpha$ satisfies all of the constraints introduced in the context of $C$.

\medskip
{\noindent\bf Guarding the first vertex in a convex region.}  First, denote $(e,i)=\how_C(1)$ and $x=\bij^{-1}(e,i)$. In addition, denote the first vertex in $C$ by $q$. Observe that $\first(s_{\how_C(1)},C)=q$, which means that $s_{\how_C(1)}$ sees $q$. If $e\in\reflex(P)$, then $e$ sees $q$ and no constraint is introduced. Next, suppose that $e\in{\cal C}(P)$. Let $\ell=\first(q,e)$ and $h=\last(q,e)$. Since $s_{\how_C(1)}$ sees $q$, neither $\ell$ nor $h$ is  $\nil$. We thus have the constraints $c^1_{(C,1)}=[x \geq \ell]$ and $c^2_{(C,1)}=[x \leq h]$. To prove that $\alpha$ satisfies them, we need to show that $\ell\leq \alpha(x)\leq h$. However, this directly follows from the fact that $\alpha(x)=s_{\how_C(1)}$ sees $q$.

\medskip
{\noindent\bf Guarding the last vertex in a convex region.}  Secondly, denote $(e,i)=\how_C(\og(C))$ and $x=\bij^{-1}(e,i)$. In addition, denote the last vertex in $C$ by $q$. Observe that $\last(s_{\how_C(\og(C))},$ $C)=q$, which means that $s_{\how_C(\og(C))}$ sees $q$. If $e\in\reflex(P)$, then $e$ sees $q$ and no constraint is introduced. Next, suppose that $e\in{\cal C}(P)$. Let $\ell=\first(q,e)$ and $h=\last(q,e)$. Since $s_{\how_C(\og(C))}$ sees $q$, neither $\ell$ nor $h$ is $\nil$. We thus have the constraints $c^1_{(C,\og(C))}=[x \geq \ell]$ and $c^2_{(C,\og(C))}=[x \leq h]$. To prove that $\alpha$ satisfies them, we need to show that $\ell\leq \alpha(x)\leq h$. However, this directly follows from the fact that $\alpha(x)=s_{\how_C(\og(C))}$ sees $q$.

\medskip
{\noindent\bf Guarding the middle vertices in a convex region.} Lastly, choose some $t\in \{2,\ldots,\og(C)\}$. Denote $(e,i)=\how_C(t)$, $x=\bij^{-1}(e,i)$, $(e',i')=\how_C(t-1)$ and $x'=\bij^{-1}(e',i')$. Note that $\alpha(x)=s_{\how_C(t)}\in e$ and ${\alpha(x')}=s_{\how_C(t-1)}\in e'$. Recall that since $S$ is a solution, we have that the vertex $a=\last(s_{\how_C(t-1)},C)$ is {\em (i)} larger or equal to $b-1$ where $b=\first(s_{\how_C(t)},C)$, and {\em (ii)} smaller than $q=\last(s_{\how_C(t)},C)$. Note that $a=\last({\alpha(x')},C)$, $b=\first({\alpha(x)},C)$ and $q=\last({\alpha(x)},C)$. This implies that $a(x)\in e$ sees at least one vertex in $C$ as well as that $a(x')\in e'$ sees at least one vertex in $C$. In particular, four constraints are introduced, and it is immediate that both $\widetilde{c}^1_{(C,t)}$ and $\widetilde{c}^2_{(C,t)}$ are satisfied.

In what follows, we need to show that $\alpha$ satisfies the constraints inserted in our two sets of four cases, which depend on the way $e'$ views $C$ with respect to $\last$, and the way $e$ views $C$ with respect to both $\first$ and $\last$. In the analysis of all cases below, when we identify $f(\alpha(x'))$, we rely on the fact that $a\neq\nil$ and $a+1\in C$ (because $a\leq q+1$ and $q\in C$). Moreover, for the first set of four cases, we rely on the fact that there exists a vertex $j\in e$ such that $\first(j,C)\leq {a+1}$ (because $b\leq a+1$). For the second set set of four cases, we rely on the fact that there exists a vertex $j\in e$ such that $\last(j,C)\geq {a+1}$ (because $a\leq q-1$). Here, the analysis of some of the cases is identical (e.g., the first and third cases of the first set); however, recall that in other proofs, these cases were analyzed differently (e.g., in the proof of monotonicity).

\medskip
\noindent{\bf Case 1 of First Set.} The way $e'$ views $C$ is non-decreasing with respect to $\last$, and the way $e$ views $C$ is non-decreasing with respect to $\first$. Let $c=[x \leq f(x')]$ be the constraint inserted in this case. To prove that $\alpha$ satisfies $c$, we need to show that $\alpha(x)\leq f(\alpha(x'))$. By the discussion before the case analysis, $f(\alpha(x'))$ is the largest vertex $j\in e$ such that $\first(j,C)\leq {a+1}$. Then, we need to show that $\alpha(x)\leq j$. However, since $\first({\alpha(x)},C)\leq {a+1}$, the inequality follows.

\medskip
\noindent{\bf Case 2 of First Set.} The way $e'$ views $C$ is non-decreasing with respect to $\last$, and the way $e$ views $C$ is non-increasing with respect to $\first$. Let $c=[x \geq f(x')]$ be the constraint inserted in this case. To prove that $\alpha$ satisfies $c$, we need to show that $\alpha(x)\geq f(\alpha(x'))$. By the discussion before the case analysis, $f(\alpha(x'))$ is the smallest vertex $j\in e$ such that $\first(j,C)\leq {a+1}$. Then, we need to show that $\alpha(x)\geq j$. However, since $\first({\alpha(x)},C)\leq {a+1}$, the inequality follows.

\medskip
\noindent{\bf Case 3 of First Set.} The way $e'$ views $C$ is non-increasing with respect to $\last$, and the way $e$ views $C$ is non-decreasing with respect to $\first$. Let $c=[x \leq f(x')]$ be the constraint inserted in this case. To prove that $\alpha$ satisfies $c$, we need to show that $\alpha(x)\leq f(\alpha(x'))$. By the discussion before the case analysis, $f(\alpha(x'))$ is the largest vertex $j\in e$ such that $\first(j,C)\leq {a+1}$. Then, we need to show that $\alpha(x)\leq j$. However, since $\first({\alpha(x)},C)\leq {a+1}$, the inequality follows.

\medskip
\noindent{\bf Case 4 of First Set.} The way $e'$ views $C$ is non-decreasing with respect to $\last$, and the way $e$ views $C$ is non-increasing with respect to $\first$. Let $c=[x \geq f(x')]$ be the constraint inserted in this case. To prove that $\alpha$ satisfies $c$, we need to show that $\alpha(x)\geq f(\alpha(x'))$. By the discussion before the case analysis, $f(\alpha(x'))$ is the smallest vertex $j\in e$ such that $\first(j,C)\leq {a+1}$. Then, we need to show that $\alpha(x)\geq j$. However, since $\first({\alpha(x)},C)\leq {a+1}$, the inequality follows.

\medskip
\noindent{\bf Case 1 of Second Set.} The ways $e'$ and $e$ view $C$ are both non-decreasing with respect to $\last$. Let $c=[x \geq f(x')]$ be the constraint inserted in this case. To prove that $\alpha$ satisfies $c$, we need to show that $\alpha(x)\geq f(\alpha(x'))$. By the discussion before the case analysis, $f(\alpha(x'))$ is the smallest vertex $j\in e$ such that $\last(j,C)\geq {a+1}$. Then, we need to show that $\alpha(x)\geq j$. However, since $\last({\alpha(x)},C)\geq {a+1}$, the inequality follows.

\medskip
\noindent{\bf Case 2 of Second Set.} The ways $e'$ and $e$ view $C$ are non-decreasing and non-increasing, respectively, with respect to $\last$. Let $c=[x \leq f(x')]$ be the constraint inserted in this case. To prove that $\alpha$ satisfies $c$, we need to show that $\alpha(x)\leq f(\alpha(x'))$. By the discussion before the case analysis, $f(\alpha(x'))$ is the largest vertex $j\in e$ such that $\last(j,C)\geq {a+1}$. Then, we need to show that $\alpha(x)\leq j$. However, since $\last({\alpha(x)},C)\geq {a+1}$, the inequality follows.

\medskip
\noindent{\bf Case 3 of Second Set.} The ways $e'$ and $e$ view $C$ are non-increasing and non-decreasing, respectively, with respect to $\last$. Let $c=[x \geq f(x')]$ be the constraint inserted in this case. To prove that $\alpha$ satisfies $c$, we need to show that $\alpha(x)\geq f(\alpha(x'))$. By the discussion before the case analysis, $f(\alpha(x'))$ is the smallest vertex $j\in e$ such that $\last(j,C)\geq {a+1}$. Then, we need to show that $\alpha(x)\geq j$. However, since $\last({\alpha(x)},C)\geq {a+1}$, the inequality follows.

\medskip
\noindent{\bf Case 4 of Second Set.} The ways $e'$ and $e$ view $C$ are both non-increasing with respect to $\last$. Let $c=[x \leq f(x')]$ be the constraint inserted in this case. To prove that $\alpha$ satisfies $c$, we need to show that $\alpha(x)\leq f(\alpha(x'))$. By the discussion before the case analysis, $f(\alpha(x'))$ is the largest vertex $j\in e$ such that $\last(j,C)\geq {a+1}$. Then, we need to show that $\alpha(x)\leq j$. However, since $\last({\alpha(x)},C)\geq {a+1}$, the inequality follows.
\end{proof}

\section{Algorithm for Monotone 2-CSP}\label{sec:algoCSP}
In this section, we design a polynomial time algorithm for {\sc Monotone 2-CSP}, running in time $\rtimedcsp$. We obtain this algorithm by reducing the given instance $(X,C,N)$ to an instance of {\sc $2$-SAT}. We note that without monotonicity or arity bound, the problem is \WOH, while when we have both these conditions (and arity is at most two), then our algorithm shows that the problem is polynomial time solvable. Indeed, to see the necessity for monotonicity, consider a reduction from {\sc Multicolored Clique} to {\sc 2-CSP} as follows. For each vertex and edge in the hypothetical solution, we create a variable. That is, we have a variable $x_i$, for each $i \in [k]$, and for every distinct $i,j \in [k]$, where $i<j$, we have a variable $e_{ij}$. We can define two functions $f^1_{ij}$ and $f^2_{ij}$ which return the vertex from $i^{th}$ and $j^{th}$ part incident to the edge $e_{ij}$, respectively. Now we add constraints of the form $x_i = f_{ij}^1(e_{ij})$ and $x_j = f_{ij}^1(e_{ji})$, for $i<j$. Notice that the selected set of vertices and edges form a clique if and only if the {\sc 2-CSP} is satisfied for the respective assignment. Critically, note that the functions that we create are not monotone. Hence, the problem is \WOH, without the monotonicity condition. The necessity for monotonicity is given by Fomin et al.~\cite{cliquewidth}, who showed that for arity $4$ and when a requirement stronger than monotonicity is imposed, the problem is \WOH. 


%

If the function $f$ of a constraint $c=[x_i \si f(x_j)]$ is constantly $\beta$ (that is, for every $t\in\{0,\ldots,N\}$, $f(t)=\beta$), then we use the shorthand $c=[x_i \si \beta]$. Moreover, we suppose that every constraint represented by a quadruple is associated with two distinct variables.

Let $(X,C,N)$ be an instance of {\sc Monotone 2-CSP}. We create a {$2$-CNF-SAT} formula $\SC{C}$ as follows (we only describe its variables and clauses). For each $x \in X$ and $d\in \{0,1,\ldots, N,N+1\}$, we create a variable $x[d]$. (Setting $x[d] = 1$ will be interpreted as $x \geq d$.) We now describe the clauses that we create. 

\myparagraph{Ensuring Valid Assignments for Variables.}  We need to ensure that $x$ is assigned some value from $\{0,1,\ldots, N\}$. Thus, for each $x \in X$, $x \geq 0$ should always be satisfied. To ensure this, we add the clause $(x[0])$ to $\SC{C}$, for every $x\in X$. Similarly, to ensure that $x \leq N$, we add the clause $(\neg x[N+1])$ to $\SC{C}$, for each $x\in X$. 

\myparagraph{Encoding Order Implications.} For each $x \in X$ and $d \in \{1,2,\ldots, N,N+1\}$, we add the clause $(x[d] \to  x[d-1])$ to $\SC{C}$. (The above clauses ensure that if $x \geq d$, then $x \geq d-1$ also holds.)

\myparagraph{Encoding constant functions}. Consider a constraint of the form $c = [x \leq \beta]$, where $\beta \in \{0,1,\ldots, N\}$. We add the clause $(\neg x[\beta + 1])$ to $\SC{C}$. Next consider a constraint of the form $c = [x \geq \beta]$, where $\beta \in \{0,1,\ldots, N\}$. (We can safely assume that $\beta < N+1$, otherwise we can correctly report that the instance is a no-instance.) We add the clause $ (x[\beta])$ to $\SC{C}$. 

\myparagraph{Encoding non-constant functions.} We encode $c = [x_i \si f(x_j)]\in C$ based on different cases of $\si \in \{ \leq, \geq\}$ and whether $f$ is non-increasing or non-decreasing. 
\begin{enumerate}
\item $\si =$ $\geq$ and $f$ is non-decreasing. For each $d \in \{0,1,\ldots, N\}$, we add the clause $(x_j[d] \to x_i[f(d)])$. 

\item $\si =$ $\geq $ and $f$ is non-increasing. For each $d \in \{0,1,\ldots, N\}$, we add the clause $(\neg x_j[d+1] \to x_i[f(d)])$. 

\item $\si =$ $\leq$ and $f$ is non-decreasing. For each $d \in \{0,1,\ldots, N\}$, we add the clause $(\neg x_j[d] \to \neg x_i[f(d) + 1])$.

\item $\si =$ $\leq $ $f$ is non-increasing. For each $d \in \{0,1,\ldots, N\}$, we add the clause $(x_j[d] \to  \neg x_i[f(d) + 1])$. 
\end{enumerate}

In the following lemma we prove the correctness of our reduction. 

\begin{lem}\label{lem:correct-reduction-2-sat}
$(X,C,N)$ is a yes-instance of {\sc Monotone 2-CSP} if and only if $\SC{C}$ is a yes-instance of {\sc $2$-SAT}. 
\end{lem}
\begin{proof}
Let $Z$ be the set of variables of $\SC{C}$. For one direction assume that $(X,C,N)$ is a yes-instance of {\sc Monotone 2-CSP}, and let $\asg: X \rightarrow \{0,1,\ldots, N\}$ be its solution. We construct an assignment $\varphi: Z \rightarrow \{0,1\}$ as follows. Consider $x\in X$ and $d \in \{0,1,\ldots, N, N+1\}$. If $\asg(x) \leq d$, then we set $\varphi(x[d]) =1$, otherwise, we set $\varphi(x[d]) = 0$. We will show that $\varphi$ is a satisfying assignment for $\SC{C}$. For $x \in X$ as $\alpha(x) \in \{0,1,\ldots, N\}$, the clauses $(x[0])$ and $(\neg x[N+1])$ are satisfied, thus the clauses ensuring valid assignments for variables and clauses for order implications are satisfied. Consider $\beta \in \{0,1,\ldots, N\}$ and $x \in X$. If $[x \leq \beta] \in C$, then by the construction of $\varphi$, the clause $(\neg x[\beta + 1])\in \SC{C}$ is satisfied. Similarly, if $[x \geq \beta] \in C$, then the clause $(x[\beta]) \in \SC{C}$ is satisfied. Thus, all the clauses encoding constant functions are satisfied. Now consider a constraint $c = [x_i \si f(x_j)] \in C$, and consider the following cases based on $\si \in \{\leq, \geq\}$ and whether $f$ is non-decreasing or non-decreasing. 

\begin{enumerate}
\item If $\si =$ $\geq$ and $f$ is non-decreasing, then for each $d \in \{0,1,\ldots, N\}$, we have the clause $(x_j[d] \to x_i[f(d)])$ in $\SC{C}$. We show that all the above clauses are satisfied by $\varphi$. Consider some $d \in \{0,1,\ldots, N\}$. If $d > \asg(x_j)$, then $\varphi(x_j[d]) = 0$, and thus $(x_j[d] \to x_i[f(d)])$ is satisfied. Now consider the case when $d \leq \asg(x_j)$. As $\alpha$ is a solution for the instance $(X,C,N)$, we have $f(\asg(x_j)) \leq \asg(x_i)$. As $f$ is non-decreasing, we have $f(d) \leq f(\asg(x_j)) \leq \asg(x_i)$. Thus we can conclude that $(x_j[d] \to x_i[f(d)])$ is satisfied by $\varphi$. 

\item If $\si =$ $\geq $ and $f$ is non-increasing, then for each $d \in \{0,1,\ldots, N\}$, we have $(\neg x_j[d+1] \to x_i[f(d)]) \in \SC{C}$. Consider some $d \in \{0,1,\ldots, N\}$. If $d < \asg(x_j)$, then $\varphi(x_j[d+1]) = 1$, and thus $(\neg x_j[d+1] \to x_i[f(d)])$ is satisfied. Now consider the case when $d \geq \asg(x_j)$, and $\varphi(x_j[d+1]) = 0$. As $\alpha$ is a solution for the instance $(X,C,N)$, we have $f(\asg(x_j)) \leq \asg(x_i)$. As $f$ is non-increasing, we have $f(d) \leq f(\asg(x_j)) \leq \asg(x_i)$. Thus we can conclude that $(\neg x_j[d+1] \to x_i[f(d)])$ is satisfied by $\varphi$.

\item If $\si =$ $\leq$ and $f$ is non-decreasing, then for each $d \in \{0,1,\ldots, N\}$, we have $(\neg x_j[d] \to \neg x_i[f(d) + 1]) \in \SC{C}$. Consider some $d \in \{0,1,\ldots, N\}$. If $d \leq \asg(x_j)$, then $\varphi(x_j[d]) = 1$, and thus $(\neg x_j[d] \to \neg x_i[f(d) + 1])$ is satisfied by $\varphi$. Now consider the case when $d > \asg(x_j)$, and $\varphi(x_j[d]) = 0$. As $\alpha$ is a solution for the instance $(X,C,N)$ and $f$ is non-decreasing, we have $\asg(x_i) \leq f(\asg(x_j)) \leq f(d)$. Thus, $\varphi(x_i[f(d) + 1]) = 0$, and we can conclude that $(\neg x_j[d] \to \neg x_i[f(d) + 1])$ is satisfied by $\varphi$.

\item If $\si =$ $\leq $ $f$ is non-increasing, for each $d \in \{0,1,\ldots, N\}$, we have $(x_j[d] \to  \neg x_i[f(d) + 1]) \in \SC{C}$. Consider some $d \in \{0,1,\ldots, N\}$. If $d > \asg(x_j)$, then $\varphi(x_j[d]) = 0$, and thus $(x_j[d] \to  \neg x_i[f(d) + 1])$ is satisfied. Now consider the case when $d \leq \asg(x_j)$, and $\varphi(x_j[d]) = 1$. As $\alpha$ is a solution for the instance $(X,C,N)$ and $f$ is non-increasing, we have $\asg(x_i) \leq f(\asg(x_j)) \leq f(d)$. Thus, $\varphi(x_i[f(d) + 1]) = 0$, and we can conclude that $(x_j[d] \to  \neg x_i[f(d) + 1])$ is satisfied by $\varphi$.
\end{enumerate}
The above discussions cover all clauses in $\SC{C}$, thus we can conclude that $\SC{C}$ is a yes-instance of {\sc $2$-SAT}.

For the other direction, let $\SC{C}$ be a yes-instance of {\sc $2$-SAT}, and let $\varphi: Z \rightarrow \{0,1\}$ be its solution. From the clauses for encoding valid assignments and order implications, for each $x\in X$, there is $d_x \in \{0,1,\ldots, N\}$, such that for all $d \in \{0,1, \ldots, d_x\}$, we have $x[d] = 1$ and for any $d' \in \{d+1, d+2, \ldots, N, N+1\}$, we have $x[d] = 0$. We construct $\asg : X \rightarrow \{0,1,\ldots, N\}$, by setting $\asg(x)=d_x$, where $x\in X$. We argue that $\asg$ is a solution for the instance $(X,C,N)$. Consider a clause of the form $[x \leq \beta] \in C$, where $\beta \in \{0,1,\ldots, N\}$. As the clause $(\neg x[\beta + 1])\in \SC{C}$ is satisfied by $\varphi$, we have $\asg(x) = d_x \leq \beta$. Thus, $[x \leq \beta] \in C$ is satisfied by $\asg$. Next, consider a clause of the form $[x \geq \beta] \in C$, for some $\beta \in \{0,1,\ldots, N\}$. As $ (x[\beta]) \in \SC{C}$ is satisfied by $\varphi$, we have $\asg(x) = d_x \geq \beta$. Now consider a constraint $c = [x_i \si f(x_j)] \in C$, and consider the following cases based on $\si \in \{\leq, \geq\}$ and whether $f$ is non-decreasing or non-increasing. 

\begin{enumerate}
\item If $\si =$ $\geq$ and $f$ is non-decreasing, then for each $d \in \{0,1,\ldots, N\}$, we have the clause $(x_j[d] \to x_i[f(d)])$ in $\SC{C}$. Note that we have $\varphi(x_j[d_{x_j}]) = 1$ and hence, $\varphi(x[f(d_{x_j})]) = 1$. Thus, $d_{x_i} \geq f(d_{x_j})$. Hence we can conclude that $\asg(x_j) =  d_{x_i} \geq f(\asg(x_j))$.

\item If $\si =$ $\geq $ and $f$ is non-increasing, then for each $d \in \{0,1,\ldots, N\}$, we have $(\neg x_j[d+1] \to x_i[f(d)]) \in \SC{C}$. Note that $\varphi(x_j[d_{x_j}+1]) = 0$. Thus $\asg(x_i) = d_x \geq f(\asg(x_j))$.

\item If $\si =$ $\leq$ and $f$ is non-decreasing, then for each $d \in \{0,1,\ldots, N\}$, we have $(\neg x_j[d] \to \neg x_i[f(d) + 1]) \in \SC{C}$. As $\varphi(x_j[d_{x_j}+1]) = 0$, we must have $d_x \leq f(d_{x_j})$. Thus, $\asg(x_i) = d_x \leq f(\asg(x_j))$.

\item If $\si =$ $\leq $ $f$ is non-increasing, for each $d \in \{0,1,\ldots, N\}$, we have $(x_j[d] \to \neg x_i[f(d) + 1]) \in \SC{C}$. As $\varphi(x_j[d_{x_j}]) = 1$, we must have $d_x \leq f(d_{x_j})$. Thus, we have $\asg(x_i) = d_x \leq f(\asg(x_j))$.
\end{enumerate}
Thus, we can conclude that $(X,C,N)$ is a yes-instance of {\sc Monotone 2-CSP}. 
\end{proof}

{\sc $2$-SAT} admits an algorithm running in time $\rtimetsat$, where $n$ is the number of variables and $m$ is the number of clauses~\cite{DBLP:journals/ipl/AspvallPT79}. This together with the construction of the {\sc $2$-SAT} instance $\SC{C}$ for the given instance $(X,C,N)$ of {\sc Monotone 2-CSP} and Lemma~\ref{lem:correct-reduction-2-sat}, implies Theorem~\ref{thm:csp}.

\section{Discretization for Boundary-Vertex Art Gallery and Vertex-Boundary Art Gallery}\label{sec:discretization}
In this section we show how we can discretize the given polygon to solve {\sc Boundary-Vertex Art Gallery} and {\sc Vertex-Boundary Art Gallery}, using the techniques used by our algorithm for {\sc Vertex-Boundary Art Gallery}. 

\begin{figure}[t]
\centering
\fbox{\includegraphics[scale=0.7]{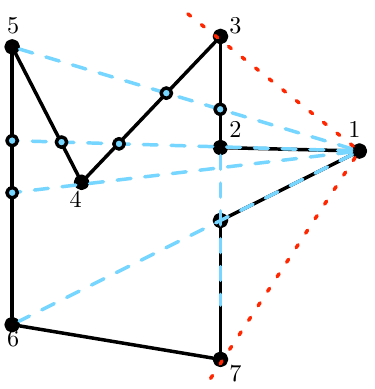}}
\caption{A (partial) illustration of the construction of $\Essn(P)$. The labelled vertices are the vertices of the polygon, whereas the blue vertices are the newly added vertices.}\label{fig:csp}
\end{figure}

We create a set $\Essn(P)$ of ``essential points'' of $P$, which will be useful for ``discretization''. 


\begin{definition}\label{def:essential-set-BB}
Consider a simple polygon $P$ with $V(P)=\{1,2,\cdots,n\}$ and $E(P) =\{\{i,i+1\} : i\in [n]\}$ (computation modulo $n$). The \emph{essential set} of $P$ is the set $\Essn(P)$ constructed as follows. Initially, $\Essn(P)$ contains all the vertices of $P$. For every distinct vertices $i,j\in [n]$, consider the line $L_{ij}$ containing $i$ and $j$. For each edge $e=\{i',j'\}$ which is not a sub-segment of $L_{ij}$, we add the intersection point (if it exists) of $L_{ij}$ and the line segment $\overline{i'j'}$, to the set $\Essn(P)$.
\end{definition} 

Note that $\Essn(P)$ can be computed in polynomial time. (We remark that by constructing $\Essn(P)$ more carefully (than what we do), we may optimize its size, but we choose to construct it this way to keep the definition simple.) 
Let $P_1$ be the polygon with vertex set $\Essn(P)$, obtained from $P$ by sub-dividing edges of $P$ (possibly multiple times). 

In the {\sc Boundary-Vertex Art Gallery} problem, the guards are placed on the boundary of $P$ and the objective is to guard the vertices of $P$. In the next lemma shows that if the given instance $(P,k)$ of {\sc Boundary-Vertex Art Gallery} is a yes-instance, then there is a solution which places guards only at vertices from $P_1$.

\begin{lem}\label{lem:place-guard-essn}
Let $(P,k)$ be a yes-instance of {\sc Boundary-Vertex Art Gallery}. Then there is a solution $S \subseteq V(P_1)$ to the instance $(P,k)$ of {\sc Boundary-Vertex Art Gallery}. 
\end{lem}
\begin{proof}
Consider a minimal solution $S$ to $(P,k)$, where $S$ is a set of points from the boundary of $P$ of size at most $k$, and $S$ is a solution that maximizes $|V(P_1) \cap S|$. We will show that $S \subseteq V(P_1)$. Towards a contradiction suppose that $S \not \subseteq V(P_1)$, and consider a point $q \in S \setminus V(P_1)$. As $q \notin V(P_1)$, there is a unique edge in $P_1$ containing it, denote that edge by $e=\{u,w\}$, where $u < w$. Let $S' = (S \setminus \{q\}) \cup \{u\}$. We will show that $S'$ is also a solution for the instance $(P,k)$, thus contradicting the choice of $S$. To prove that $S'$ is a solution, it is enough to show that for every $v \in V(P)$ that is seen by $q$, $u$ also sees $v$. Consider some $v \in V(P)$ that is seen by $q$. Towards a contradiction assume that $u$ does not see $v$. Let $T$ be the triangle defined by $v,u$ and $q$. As $u$ does not see $v$ and $q \notin V(P_1)$, $T$ is a non-degenerate triangle. Also the line segment $\overline{uv}$ is not completely contained in $P$ (or $P_1$), and thus there is a reflex vertex $v^*$ from $P$ that is either strictly contained inside $T$ or contained in the line segment $\overline{vq}$. In either case, the line $L$ containing $v$ and $v^*$ intersects $\overline{uq}$ at a point different than $u$. This contradicts that $\{u,w\}$ is the edge in $P_1$ containing $q$, where $q\notin V(P_1)$. This concludes the proof. 
\end{proof}

We now briefly explain how we can obtain an \FPT\ algorithm for {\sc Boundary-Vertex Art Gallery} using the techniques that we used in Section~\ref{sec:art} and Lemma~\ref{lem:place-guard-essn}. Let $(P,k)$ be an instance of {\sc Boundary-Vertex Art Gallery}, and define $P_1$ as was described earlier. The first component of our algorithm for {\sc Vertex-Vertex Art Gallery} was a Turing reduction to a structured form of {\sc Art Gallery}, called {\sc Structured Art Gallery} (see Section~\ref{sec:structured}). We can define a {\sc Structured Boundary-Vertex Art Gallery} which takes an additional input, which is the set of vertices to be guarded. In additional to all other inputs, we provide $P_1$ as the input polygon and $V(P) \subseteq V(P_1)$ as the set of vertices to be guarded. The safeness of the above Turing reduction can be obtained from Lemma~\ref{lem:place-guard-essn} and arguments similar to the one used for the proof of Lemma~\ref{lem:correct-Turing-reduction-vvag}. The next step is to reduce the structured instance to an instance of {\sc Monotone CSP}. We follow similar procedure as given in Section~\ref{sec:reduction}, but we restrict the ranges for the functions to vertices appearing in $V(P)$. Finally, we resolve the instance by solving the instances of {\sc Monotone CSP}, using Theorem~\ref{thm:csp}. From the above discussions we can obtain the following theorem.

\begin{theorem}\label{thm:boundary-vertex-art-gallery}
{\sc Boundary-Vertex Art Gallery} is \FPT\ parameterized by $r$, the number of reflex vertices. In particular, it admits an algorithm with running time $r^{\OO(r^2)}n^{\OO(1)}$.
\end{theorem}

Next we turn to {\sc Vertex-Boundary Art Gallery}. Recall that in the {\sc Vertex-Boundary Art Gallery} problem, the guards are to be placed on the vertices of $P$ and the goal is to guard the whole boundary of $P$. We obtain $P_1$ from $P$ as was described earlier. Furthermore, we obtain $P_2$ from $P_1$ by sub-dividing each edge of $P_1$ exactly once. In the next lemma we show that any set that guards all vertices of $P_2$, guards the whole boundary of $P$.

\begin{lem}\label{lem:enough-to-guard-vp2}
Let $(P,k)$ be an instance of {\sc Vertex-Boundary Art Gallery}. Consider a set $S \subseteq V(P)$ of size at most $k$, such that for each $v \in V(P_2)$, there is $s\in S$ that sees $v$. Then $S$ is a solution to the instance $(P,k)$ of {\sc Vertex-Boundary Art Gallery}. 
\end{lem}
\begin{proof}
Consider a point $p$ in the boundary of $P$ which is not a vertex of $P_2$. Let $\{u,w\}$ be the edge in $P_1$ that contains $p$ strictly in its interior. By the construction of $P_2$, there is a vertex $v \in V(P_2) \setminus V(P_1)$ contained strictly inside the line segment $\overline{uw}$. Consider $s\in S$ such that $s$ sees $v$. We will show that $s$ sees $p$. Towards a contradiction, suppose that $s$ does not see $p$. Consider the triangle $T$ formed by $p,v$ and $s$. As $s$ does not see $p$, we can conclude that $T$ is non-degenerate and $\overline{ps}$ is not completely contained in $P$. Thus, there is a reflex vertex $\what{v}$ which is either strictly contained inside $T$, or contained in the line segment $\overline{sv}$. If $\what{v}$ is strictly contained in the interior of $T$, then we can contradict that $\{u,w\}$ is the edge in $P_1$ containing $p$. Otherwise, if $\what{v}$ is contained in the line segment $\overline{sv}$, and we can obtain a contradiction to the fact that $v \in V(P_2)\setminus V(P_1)$. Thus, we obtain that $s$ sees $p$. This concludes the proof.
\end{proof}

Now we explain how we can obtain an \FPT\ algorithm for {\sc Vertex-Boundary Art Gallery} using the techniques that we used in Section~\ref{sec:art} and Lemma~\ref{lem:enough-to-guard-vp2}. Let $(P,k)$ be an instance of {\sc Vertex-Boundary Art Gallery}, and define $P_1$ and $P_2$, as was described earlier. Again we define a structured form of the problem called {\sc Structured Boundary-Vertex Art Gallery}, which takes an additional set of vertices from which the guards can be selected. We give Turing reduction from {\sc Vertex-Boundary Art Gallery} to {\sc Structured Boundary-Vertex Art Gallery}, where apart from the other inputs, the input polygon is $P_2$ and the set from which we are allowed to select guards is $V(P)$. We can obtain the correctness of the above Turing reduction using Lemma~\ref{lem:enough-to-guard-vp2} and arguments similar to the one used for the proof of Lemma~\ref{lem:correct-Turing-reduction-vvag}. The next step is to reduce the structured instance to an instance of {\sc Monotone CSP}. We follow similar procedure as given in Section~\ref{sec:reduction}, but this time we restrict the domains for the functions to vertices appearing in $V(P)$. Finally, we resolve the instance by solving the instances of {\sc Monotone CSP}, using Theorem~\ref{thm:csp}. From the above discussions we can obtain the following theorem.

\begin{theorem}\label{thm:vertex-boundary-art-gallery}
{\sc Vertex-Boundary Art Gallery} is \FPT\ parameterized by $r$, the number of reflex vertices. In particular, it admits an algorithm with running time $r^{\OO(r^2)}n^{\OO(1)}$.
\end{theorem}

\bibliography{main-AGP}

\end{document}